\documentclass[sigplan,nonacm]{acmart}
\setcopyright{none}
\renewcommand\footnotetextcopyrightpermission[1]{}

\usepackage{amsmath,amsfonts}

\usepackage{amssymb}
\usepackage{graphicx}
\usepackage{textcomp}
\usepackage{xcolor}
\usepackage{algpseudocode}
\usepackage{algorithm}
\usepackage{tikz}
\usepackage{booktabs}
\usepackage{caption}
\usepackage{enumitem}
\usepackage{mdframed}
\usepackage{multirow}
\usepackage{cancel}

\usepackage{quantikz}

\def\BibTeX{{\rm B\kern-.05em{\sc i\kern-.025em b}\kern-.08em
    T\kern-.1667em\lower.7ex\hbox{E}\kern-.125emX}}
    
\begin{document}

\title{Linear Complexity Fermionic Simulation on Quantum Devices with Hardware Connectivity Constraints}
\subtitle{Based on the version submitted for peer review in April 2026, with minor revisions.}

\author{Xiangyu Gao}
\authornote{These authors contributed equally to this work.}
\affiliation{
  \institution{University of Washington}
  \city{Seattle}
  \state{WA}
  \country{USA}
}
\email{xiangyug@cs.washington.edu}

\author{Winston Li}
\authornotemark[1]
\affiliation{
  \institution{Rutgers University}
  \city{Piscataway}
  \state{NJ}
  \country{USA}
}
\email{wl605@scarletmail.rutgers.edu}

\author{Jiakang Li}
\affiliation{
  \institution{Rutgers University}
  \city{Piscataway}
  \state{NJ}
  \country{USA}
}
\email{jiakang.li@rutgers.edu}

\author{Zirui Li}
\affiliation{
  \institution{Rutgers University}
  \city{Piscataway}
  \state{NJ}
  \country{USA}
}
\email{zirui.li@rutgers.edu}

\author{Yipeng Huang}
\affiliation{
  \institution{Rutgers University}
  \city{Piscataway}
  \state{NJ}
  \country{USA}
}
\email{yipeng.huang@rutgers.edu}

\author{Costin Iancu}
\affiliation{
  \institution{Lawrence Berkeley National Laboratory}
  \city{Berkeley}
  \state{CA}
  \country{USA}
}
\email{cciancu@lbl.gov}

\author{Eddy Z. Zhang}
\affiliation{
  \institution{Rutgers University}
  \city{Piscataway}
  \state{NJ}
  \country{USA}
}
\email{eddy.zhengzhang@gmail.com}

\newcommand{\cut}[1]{}
\newcommand{\eg}{e.g., }
\newcommand{\ie}{i.e., }
\newcommand{\etal}{{et al.}\xspace}
\newcommand{\vs}{{vs.} }
\newcommand{\etc}{{etc.}\xspace}
\newenvironment{parafont}{\fontfamily{ptm}\selectfont}{}
\newcommand{\Para}[1]{\par\vspace{1pt}\noindent\begin{parafont}\textbf{\textit{#1}}\end{parafont}}

\newcommand{\Sec}[1]{\S\ref{sec:#1}}
\newcommand{\Fig}[1]{Fig.\ref{fig:#1}}
\newcommand{\ct}{\small \tt }
\newcommand{\nop}[1]{}

\newcommand{\xiangyu}[1]{\textcolor{blue}{[XG: #1]}}
\newcommand{\winston}[1]{\textcolor{red}{[WL: #1]}}
\newcommand{\eddy}[1]{\textcolor{green}{[Eddy: #1]}}

\newcommand{\fixit}[1]{{\color{red}{#1}}}

\newcommand{\squishlist}{
   \begin{list}{$\bullet$}
    { \setlength{\itemsep}{0pt}      \setlength{\parsep}{3pt}
      \setlength{\topsep}{3pt}       \setlength{\partopsep}{0pt}
      \setlength{\leftmargin}{3.5mm} \setlength{\labelwidth}{1em}
      \setlength{\labelsep}{0.5em} } }

\newcommand{\squishend}{
    \end{list}  }
\newcommand{\subsubsubsection}[1]{\paragraph{#1}}
\newcommand{\sysname}{Accordion}

\begin{abstract}
Simulating fermionic systems on quantum hardware requires compiling
fermionic Hamiltonians into executable quantum circuits. Existing
approaches treat each compilation stage independently, applying
heuristics with localized objectives that produce circuits with
superquartic gate count and depth scaling and compilation times
reaching several hours for large instances.  We present \sysname{}, an
end-to-end framework that co-designs the fermion-to-qubit mapping with
circuit synthesis and hardware routing. \sysname{} fixes the
Jordan--Wigner mapping, which despite its higher Pauli weight produces
Pauli operators with structural regularity that enables provably
efficient circuit generation. For full-rank all-to-all electronic
structure Hamiltonians, we prove $O(N^4)$ gate count and circuit
depth, matching the information-theoretic lower bound imposed by the
$\Theta(N^4)$ second excitation terms.  On linear, IBM heavy-hex,
and square grid architectures, \sysname{} reduces gate count by
up to $79\%$ and circuit depth by up to $77\%$ relative to the best baseline.
\end{abstract}

\maketitle

\section{Introduction}
\label{sec:intro}

Simulation of fermionic systems is fundamental to quantum chemistry and
condensed-matter physics. Computing ground-state energies of molecular
Hamiltonians enables applications ranging from drug discovery to materials
design. Variational quantum eigensolvers (VQE) have emerged as a leading
approach for simulating such systems on modern quantum hardware by
variationally approximating the ground state of a fermionic Hamiltonian.

Executing VQE workloads on quantum hardware requires transforming
high-level fermionic Hamiltonians into hardware-executable quantum
circuits through a multi-stage compilation pipeline. First, fermionic
operators are encoded as sums of Pauli strings. These Pauli operators are
then synthesized into logical quantum circuits, which are further compiled
into physical circuits satisfying the device's connectivity constraints.

Prior work addresses each stage of this pipeline in isolation. At the
encoding stage, various fermion-to-qubit mappings---including
Jordan--Wigner (JW)~\cite{JW}, Bravyi--Kitaev (BK)~\cite{BK},
HATT~\cite{hatt}, and Fermihedral~\cite{Fermihedral}---have been proposed
to optimize intermediate metrics such as Pauli weight. At the circuit
synthesis stage, tools such as Tetris~\cite{Tetris} and
Paulihedral~\cite{Paulihedral} exploit structural similarities among Pauli
strings to reduce gate count and circuit depth via heuristics.
Figure~\ref{fig:CircuitSynthWF} illustrates that the space of possible
compilation paths is the cross product of mapping strategies and circuit
synthesis methods.

\begin{figure}[t]
    \centering
    \includegraphics[width=0.49\textwidth]{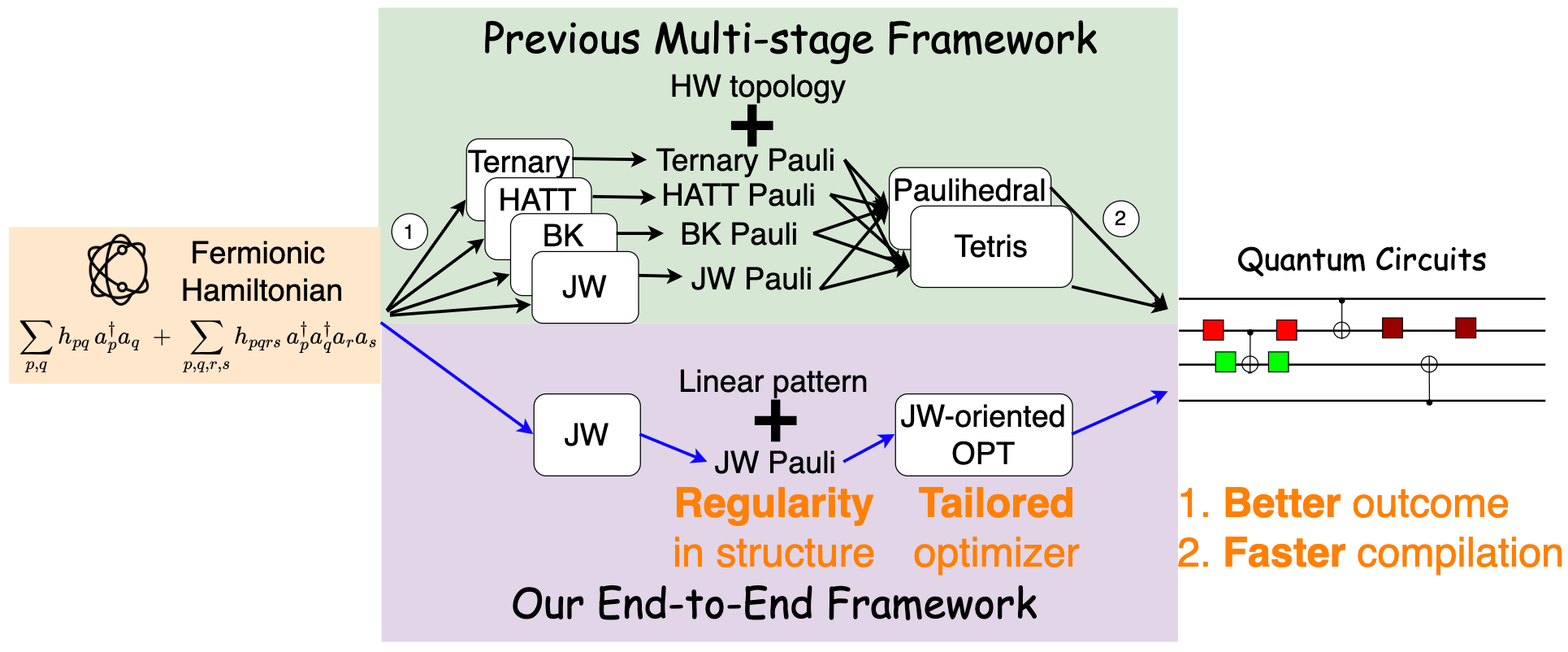}
    \vspace{-0.3in}
    \caption{State-of-the-art approaches generate quantum circuits in two
    stages: \textcircled{1} map fermionic operators to qubit operators, and
    \textcircled{2} generate hardware circuits using heuristics. \sysname{}
    takes an end-to-end approach by fixing the mapping and designing
    tailored circuit generation algorithms.}
    \label{fig:CircuitSynthWF}
\end{figure}

However, optimizing these stages in isolation is fundamentally limited. The
intermediate metrics that individual stages optimize---most notably Pauli
weight---do not reliably predict the quality of the final compiled
circuit. For instance, the BK mapping reduces Pauli weight to
$O(\log N)$ compared to JW's $O(N)$, yet the resulting Pauli strings
exhibit less structural regularity. Because the effectiveness of
downstream circuit synthesis techniques depends heavily on exploiting
structural patterns among Pauli strings, BK's lower Pauli weight can
paradoxically yield higher gate count and circuit depth in the final
circuit. Furthermore, as problem size grows, all existing combinations
of mappers and synthesizers exhibit superquartic scaling: both gate count
and circuit depth grow substantially faster than the number of Pauli
strings (Figure~\ref{fig:ScalabilityA2A}), and compilation times reach
several hours for large benchmarks.

The ultimate measure of compilation quality is end-to-end circuit
efficiency---gate count and circuit depth on the target hardware---not
intermediate metrics such as Pauli weight. We therefore argue that the
compilation pipeline should be designed and optimized holistically.

In this paper, we present \sysname{}, an end-to-end framework that
compiles electronic structure fermionic Hamiltonians directly to hardware
circuits. Rather than treating the mapping and synthesis stages
independently, \sysname{} fixes the JW mapping. This produces Pauli
operators with a provably regular, predictable structure, and develops
circuit synthesis and hardware routing algorithms that are specifically
tailored to exploit that structure. The key insight is that JW's
structural regularity enables a class of gate cancellation and qubit
scheduling optimizations that are not accessible to general-purpose
synthesizers operating on the less uniform outputs of BK or other
mappings. \sysname{} further exploits the observation that a linear
connectivity path spanning most or all qubits is a common feature across
modern quantum hardware topologies, providing a unified compilation
substrate for diverse devices.

The major contributions of this paper are as follows.

\begin{itemize}

    \item \textbf{End-to-end co-design framework.} We build \sysname{},
    the first end-to-end framework for compiling electronic structure
    fermionic Hamiltonians that co-designs the fermion-to-qubit mapping
    with the circuit synthesis and hardware routing stages. By fixing the
    JW mapping and specializing all downstream compilation to its
    structured output, \sysname{} achieves optimizations that
    stage-isolated approaches cannot.

    \item \textbf{Counter-intuitive mapping choice with provable benefits.}
    We demonstrate that contrary to conventional wisdom, the JW mapping
    outperforms BK, HATT, and balanced ternary tree approaches
    in end-to-end circuit quality when paired with the tailored synthesis
    techniques developed in this paper.
    This is surprising, since the JW mapping has the highest average Pauli weight among all standard
    mappings.
    This result reveals that
    minimizing Pauli weight is not a reliable proxy for minimizing
    hardware circuit cost.

    \item \textbf{Pauli string permutation and scheduling algorithms with
    provable linear overhead.} We develop Pauli string grouping,
    permutation, and scheduling algorithms that maximize CNOT gate
    cancellation across strings and minimize SWAP overhead during hardware
    routing. We prove that transitions between adjacent string groups
    require only a constant number of SWAP operations, yielding $O(1)$
    amortized depth per Pauli string.

    \item \textbf{Proven $O(N^4)$ circuit complexity for all-to-all
    Hamiltonians.} We formally prove that \sysname{} produces circuits
    with $O(N^4)$ gate count and circuit depth for full-rank all-to-all
    electronic structure Hamiltonians, matching the information-theoretic
    lower bound imposed by the $\Theta(N^4)$ number of second excitation
    terms. No existing approach achieves this bound.

    \item \textbf{Substantial empirical improvements.} We evaluate
    \sysname{} against state-of-the-art baselines on representative
    quantum chemistry benchmarks across NISQ (IBM heavy-hex) and
    fault-tolerant (square grid) architectures. \sysname{} reduces gate
    count by up to $79.56\%$ (avg. $44.96\%$), reduces circuit depth by
    up to $77.24\%$ (avg. $51.20\%$), while increasing compile time
    by at most $87.93\%$ (avg. $15.89\%$) relative to the best available
    baseline.

\end{itemize}

\section{Background}
\label{background}

In this section, we introduce essential background on fermionic
Hamiltonians and the current quantum circuit synthesis pipeline.
Foundational concepts in quantum computing are covered in standard
references~\cite{QCQI}.

\subsection{Fermionic Hamiltonians}

Fermionic Hamiltonians are the standard mathematical model for
interacting many-body fermionic systems. They arise in electronic
structure problems in quantum chemistry~\cite{modernQChem} and in
lattice models for condensed-matter physics~\cite{QTheoryManyParticle,
CondensedMFT}, where the primary computational target is typically the
ground state and ground-state energy, which determine equilibrium
properties of molecules and materials~\cite{aspuru2005simulated,
cao2019quantum, google2020hartree, stanisic2022observing}.

Several types of fermionic Hamiltonians arise in quantum simulation.
Electronic structure Hamiltonians~\cite{mcquarrie2008quantum} model
molecular systems; the Fermi--Hubbard model~\cite{CondensedMFT} models
condensed-matter systems. Among these, the electronic structure
Hamiltonian is of particular importance because it directly captures the
many-electron problem underlying chemistry and materials science, and
serves as the primary target for quantum simulation algorithms. The
electronic Hamiltonian takes the form

\begin{equation}
H = \sum_{p,q} h_{pq}\,a_p^\dagger a_q
\;+\;\tfrac{1}{2}\sum_{p,q,r,s} h_{pqrs}\,a_p^\dagger a_q^\dagger a_r a_s,
\label{eq:elec-hamiltonian}
\end{equation}

\noindent where $a_p^\dagger$ and $a_p$ are fermionic creation and
annihilation operators, and the coefficients $h_{pq}$ and $h_{pqrs}$
are one- and two-electron integrals that vary by molecule.

To compute the ground-state energy, the Variational Quantum Eigensolver
(VQE)~\cite{peruzzo2014variational} formulates the problem as a
variational optimization over parameterized quantum circuits. The
structure of these circuits is determined by the choice of ansatz.
Among the various proposals~\cite{mcclean2016theory, tilly2022variational},
the Unitary Coupled Cluster Singles and Doubles (UCCSD)
ansatz~\cite{uccsd} is one of the most widely used, and its operator
structure directly mirrors Equation~(\ref{eq:elec-hamiltonian}): single
excitations $a_p^\dagger a_q$ correspond to one-body terms, and double
excitations $a_p^\dagger a_q^\dagger a_r a_s$ correspond to two-body
terms. This paper focuses on the Hamiltonian form in
Equation~(\ref{eq:elec-hamiltonian}) and its full compilation from
fermionic operators to hardware circuits, with the UCCSD ansatz as the
primary benchmark workload.

\subsection{The Compilation Pipeline}

Because simulating fermionic systems on classical computers requires
exponential resources, quantum hardware provides a natural execution
platform. Existing approaches follow a two-stage compilation pipeline
(Figure~\ref{fig:CircuitSynthWF}).

The \emph{fermion-to-qubit} stage encodes fermionic operators as Pauli
operators, representing the Hamiltonian as a sum of Pauli strings. The
central figure of merit at this stage is Pauli weight: the number of
non-identity factors in each Pauli string. Table~\ref{tab:mapping}
summarizes the main mapping strategies. Jordan--Wigner~\cite{JW} is the
most direct mapping, producing Pauli strings with regular, predictable
structure at the cost of $O(N)$ Pauli weight. Bravyi--Kitaev~\cite{BK}
reduces Pauli weight to $O(\log N)$ via a tree-based encoding, but at
the cost of less regular operator patterns. HATT~\cite{hatt} constructs
an adaptive ternary tree tuned to the input Hamiltonian.
Fermihedral~\cite{Fermihedral} uses a SAT solver to find the
globally minimum Pauli weight mapping, but does not scale to large
problems. As we discuss in Section~\ref{sec:motivation}, structural
regularity among Pauli strings---which only the JW mapping reliably
provides---is ultimately more important for end-to-end circuit quality
than Pauli weight alone.

The \emph{qubit-to-circuit} stage compiles Pauli operators into quantum
circuits executable on hardware. This involves decomposing exponentials
of Pauli strings into gate sequences and routing them to satisfy
hardware connectivity constraints. Representative approaches---Tetris
\cite{Tetris} and Paulihedral \cite{Paulihedral}---use heuristics to
reduce gate count and circuit depth by exploiting gate cancellation and
minimizing SWAP routing overhead. The effectiveness and compilation time
of these methods depend strongly on the structure of their input Pauli
strings and the target architecture. For large problems, compilation
time can grow substantially as the heuristics explore increasingly large
search spaces.

\begin{table}[t]
\begin{small}
\caption{Fermion-to-qubit mappings and their tradeoffs.}
\label{tab:mapping}
\end{small}
\resizebox{\columnwidth}{!}{
\begin{tabular}{lcc}
\toprule
\textbf{Mapping} & \textbf{Advantage} & \textbf{Limitation} \\
\midrule
JW~\cite{JW}          & Regular operator structure   & High Pauli weight \\
BK~\cite{BK}          & Low Pauli weight             & Irregular operator structure \\
HATT~\cite{hatt}      & Input-adaptive mapping       & Less regular structure \\
Fermihedral~\cite{Fermihedral} & Minimum Pauli weight & Limited scalability \\
\bottomrule
\end{tabular}}
\end{table}

\section{Motivation}
\label{sec:motivation}

The end user ultimately cares about end-to-end circuit quality: gate
count and circuit depth on the target hardware. We examine the
state-of-the-art compilation workflow and identify four observations
that motivate our approach.

\begin{figure}[t]
    \centering
    \includegraphics[width=0.45\textwidth]{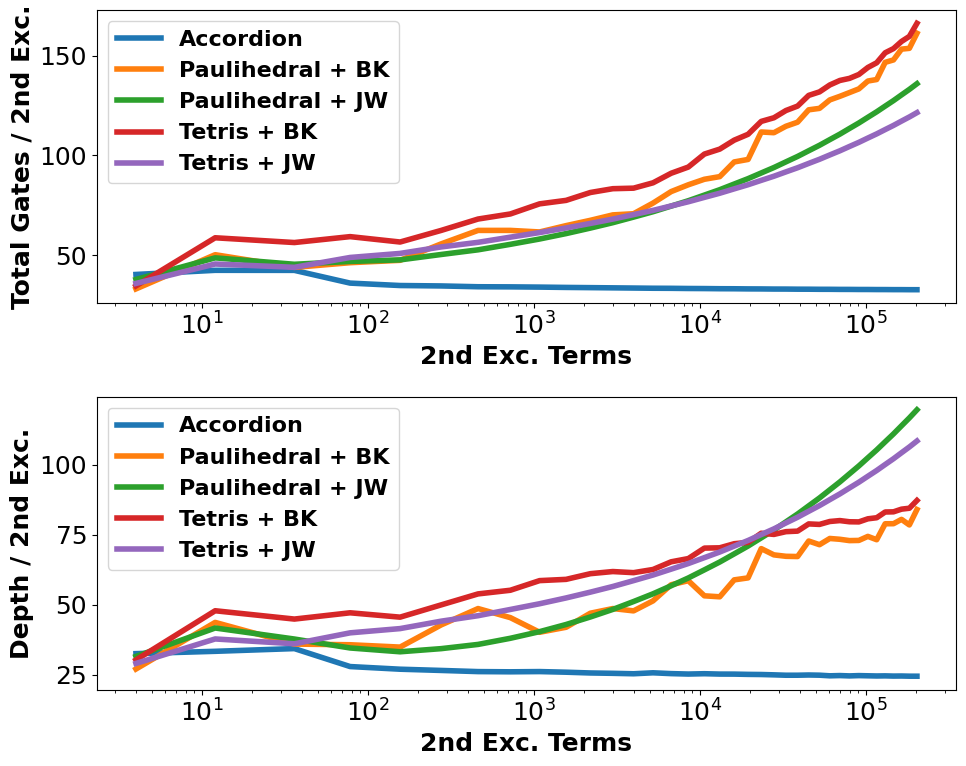}
    \vspace{-0.1in}
    \caption{Gate count and circuit depth normalized by the number of
    second excitation terms, for all-to-all Hamiltonians on a linear
    architecture. All state-of-the-art approaches exhibit superlinear
    scaling relative to the number of second excitation terms.
    \sysname{} achieves near-constant normalized cost and outperforms all
    baselines for $N \ge 11$ qubits.}
    \label{fig:ScalabilityA2A}
\end{figure}

\noindent\textbf{Observation 1: Lower Pauli weight does not imply lower
gate count or circuit depth.}
Existing pipelines optimize each stage independently, with the
fermion-to-qubit stage targeting Pauli weight under the assumption that
lower weight implies lower hardware cost. This assumption does not hold
in practice. Figure~\ref{fig:misc_pauli_weight} shows that for several
UCCSD benchmarks, the BK mapping---which produces lower Pauli weight
than JW---yields a larger final gate count than JW when both are
compiled with the same synthesizer (Paulihedral). This occurs because
the circuit synthesis techniques in the second stage rely on structural
patterns among Pauli strings to achieve gate cancellation; BK's less
regular operator structure reduces the effectiveness of these
techniques, negating its Pauli weight advantage.

Additionally, Figure~\ref{fig:ScalabilityA2A} shows that as the number
of qubits $N$ increases, all existing mapper--synthesizer combinations
produce gate count and circuit depth that grow significantly faster than
the number of second excitation terms, which scales as $O(N^4)$.

\begin{figure}[t]
    \centering
    \includegraphics[width=0.45\textwidth]{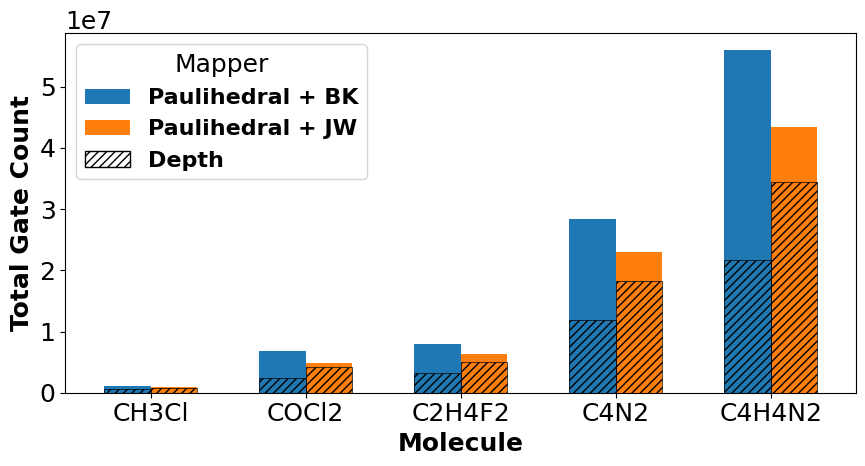}
    \vspace{-0.1in}
    \caption{UCCSD benchmarks on a linear architecture. Despite producing
    Pauli strings with lower Pauli weight, BK yields larger final gate
    count than JW on several benchmarks when compiled with Paulihedral.}
    \label{fig:misc_pauli_weight}
\end{figure}

\begin{figure}[t]
    \centering
    \includegraphics[width=0.49\textwidth]{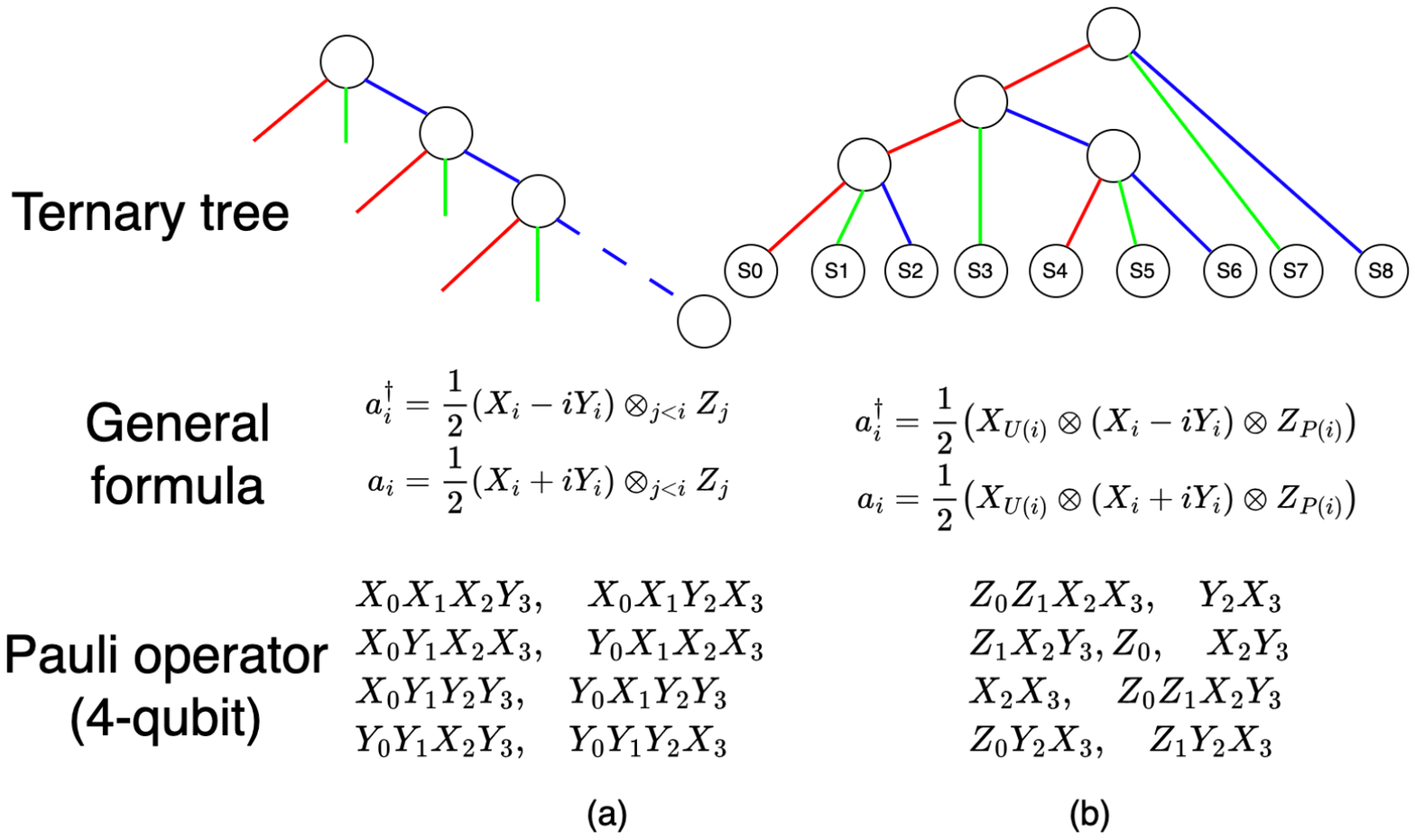}
    \vspace{-0.4in}
    \caption{Comparison of double-excitation Pauli operators generated by
    JW and BK for a 4-qubit example. (a) JW uses a linear (degenerate)
    ternary tree and produces Pauli strings with a regular, predictable
    pattern. (b) BK uses a tree-based parity encoding and produces less
    uniform Pauli strings.}
    \label{fig:BK}
\end{figure}

\noindent\textbf{Observation 2: The JW mapping produces Pauli operators
with more regular structure.}
Different mappings encode fermionic information into Pauli operators
using different strategies. In JW, each fermionic operator maps to a
Pauli string with a linear prefix structure, resulting in regular,
predictable patterns across all output strings. In BK, encoding relies
on tree-based parity (e.g., Fenwick trees), so the operator structure
depends on the specific tree topology and node indexing, yielding less
uniform patterns. Figure~\ref{fig:BK} illustrates this for a 4-qubit
example: JW produces Pauli strings that share a consistent positional
structure, whereas BK strings are more scattered and irregular. This
regularity is the key property that \sysname{} exploits.

\noindent\textbf{Observation 3: Linear connectivity is a widely shared
feature of quantum hardware and provides a useful compilation substrate.}
Although modern quantum devices exhibit a variety of hardware
topologies, a linear path spanning most or all qubits is a common
structural feature across architectures (detailed in
Appendix~\ref{appendix:linear-pattern}). Moreover, a linear path
represents the most restricted form of connectivity: any compilation
strategy that succeeds on a linear architecture will be valid on any
richer topology. Existing circuit synthesizers handle arbitrary target
topologies through general-purpose heuristics, but this generality
prevents them from exploiting the specialized optimizations that become
available when a linear structure is assumed. The contiguous Z-prefix
structure of JW's Pauli operators maps naturally onto a linear qubit
layout, which motivates our choice to pair JW with linear-connectivity
compilation as the foundation of \sysname{}.

\noindent\textbf{Observation 4: The lack of vertical integration across
compilation stages creates unavoidable suboptimalities.}
Because each stage of the existing pipeline is designed to accept
arbitrary outputs from the preceding stage, individual stages must
remain general-purpose. This generality prevents them from exploiting
structural properties of specific mappings, and forces later stages to
use heuristics capable of handling diverse inputs. The result is that no
single mapper--synthesizer combination dominates all others across
benchmarks: as shown in Figure~\ref{fig:misc_comparison}, the
combination that performs best varies by molecule, so users must test
all combinations to find the optimal one. A vertically integrated
pipeline that fixes the mapping and specializes all downstream
compilation accordingly can avoid this search and unlock global
optimization opportunities.

\begin{figure}[t]
    \centering
    \includegraphics[width=0.45\textwidth]{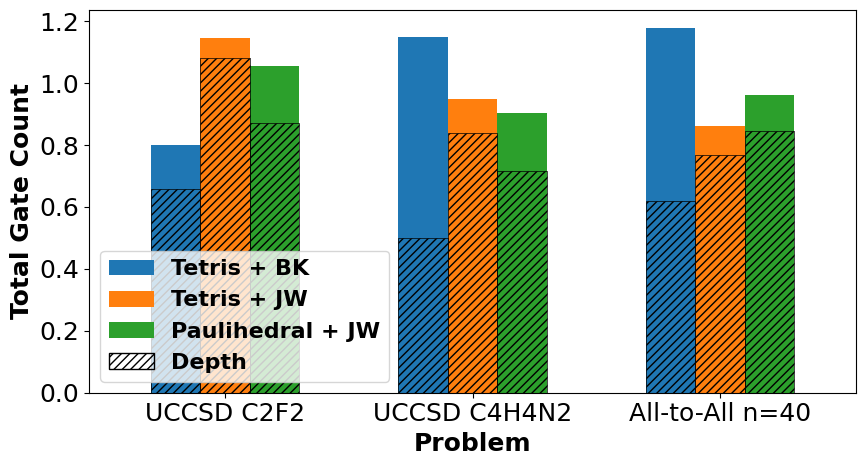}
    \caption{Total gate count for various UCCSD benchmarks on a linear
    architecture, normalized to the average across all compilation
    pipelines. No single combination of mapper and synthesizer is best
    across all benchmarks.}
    \label{fig:misc_comparison}
\end{figure}

\noindent\textbf{Takeaway.}
These observations together motivate a fundamentally different design
philosophy: rather than optimizing each stage in isolation against
intermediate metrics, one should fix a mapping that produces
structurally regular Pauli strings---namely JW---and co-design the
downstream synthesis and routing algorithms to exploit that regularity.
Our empirical results confirm that this approach achieves $O(N^4)$
circuit complexity, matching the information-theoretic lower bound for
full-rank electronic structure Hamiltonians, while all existing
approaches exhibit superquartic scaling in practice
(Appendix~\ref{appendix:scalability}).

\section{\sysname{} Implementation}
\label{approach}

This section presents \sysname{}'s end-to-end compilation approach for
electronic Hamiltonians, from fermionic operators to hardware circuits.
We fix JW mapping and develop specialized circuit synthesis algorithms
tailored to its output, targeting architectures with linear connectivity.
Our theoretical analysis focuses on all-to-all Hamiltonians, which
represent the full-rank case and provide a clean setting for complexity
analysis. We first describe Pauli string grouping
(\textsection\ref{sec:pauligroup}), then present intra-group
(\textsection\ref{sec:intra-atomic}) and inter-group
(\textsection\ref{sec:inter-atomic}) scheduling algorithms.

The electronic Hamiltonian in Equation~(\ref{eq:elec-hamiltonian})
consists of first excitation terms ($\sum_{p,q} h_{pq}\,a_p^\dagger
a_q$) and second excitation terms ($\frac{1}{2}\sum_{p,q,r,s}
h_{pqrs}\,a_p^\dagger a_q^\dagger a_r a_s$). The dominant compilation
cost comes from the second excitation terms, which under JW mapping
produce Pauli strings following the pattern $I^*AZ^*AI^*AZ^*AI^*$,
where $A \in \{X, Y\}$ and the total number of $X$ and $Y$ operators is
odd. This is made explicit in Equation~(\ref{eq:jw-pauli}), adapted from
Appendix A of \citet{uccsd}. The techniques developed below apply to
first excitation terms as well, and Section~\ref{sec:eval} confirms
their effectiveness on practical chemistry instances where not all
Hamiltonian terms are present.

\begin{align}
\label{eq:jw-pauli}
&\bigotimes_{b=l+1}^{k-1}Z_b
\bigotimes_{a=j+1}^{i-1}Z_a
\bigl(X_lX_jY_kX_i
+Y_lX_kY_jY_i
+X_lY_kY_jY_i \notag \\
+&X_lX_kX_jY_i
-Y_lX_kX_jX_i
-X_lY_kX_jX_i
-Y_lY_kY_jX_i
-Y_lY_kX_jY_i\bigr) \notag \\
=&\sum_{
\substack{
A_l,A_k,A_j,A_i\in\{X,Y\} \\
\text{\#X and \#Y both odd}
}}
\pm\,
I_1\cdots I_{l-1}
\otimes A_l\otimes
Z_{l+1}\cdots Z_{k-1}
\otimes A_k \notag \\
&\otimes
I_{k+1}\cdots I_{j-1}
\otimes A_j\otimes
Z_{j+1}\cdots Z_{i-1}
\otimes A_i\otimes
I_{i+1}\cdots I_n
\end{align}

\subsection{Pauli String Grouping}
\label{sec:pauligroup}

The JW mapping produces Pauli strings that share structural features
determined by the positions of their non-identity operators $A \in
\{X, Y\}$. We exploit this by partitioning strings into groups based on
which subset of the four $A$-operator positions they share, from most to
least constrained.

A \textbf{Singleton Group} contains a single Pauli string: all four
non-identity operators appear at a fixed set of positions. An
\textbf{Atomic Group} contains eight strings whose four $A$-operator
positions are all fixed but whose $X$/$Y$ assignments vary; these are
exactly the eight sign variants of one Pauli exponential term. A
\textbf{Mini Group} contains strings sharing the positions of the 1st,
3rd, and 4th $A$ operators, with the 2nd $A$ operator varying within a
contiguous range. A \textbf{Medium Group} contains strings sharing only
the positions of the 1st and 4th $A$ operators. A \textbf{Large Group}
contains all strings sharing the position of the 1st $A$ operator only.
Table~\ref{table:GroupDef} gives the formal definition, group counts,
and sizes for an $N$-qubit system.

\begin{table*}[t]
\renewcommand{\arraystretch}{1.5}
\caption{Pauli string group definitions, sizes, and examples for an
$N$-qubit system. $A \in \{X, Y\}$ denotes a non-identity operator.
$\text{pos}(\cdot)$ returns the qubit index of a given operator.
$\binom{m}{n}$ denotes the binomial coefficient.}
\label{table:GroupDef}
\resizebox{\linewidth}{!}{
\begin{tabular}{llccc}
\toprule
 & \textbf{Definition} & \textbf{\# groups}
 & \textbf{\# terms per group} & \textbf{Example} \\
\midrule
\textbf{Singleton} &
  All four $A$-operator positions fixed; single term. &
  $\binom{N}{4}\times 8$ & 1 &
  \texttt{\{XXIXY\}} \\
\textbf{Atomic} &
  All four $A$-operator positions fixed; $X/Y$ assignments vary. &
  $\binom{N}{4}$ & 8 &
  \texttt{\{XXIXY, XYIYY, \ldots\}} \\
\textbf{Mini} &
  Positions of 1st, 3rd, and 4th $A$ fixed; 2nd $A$ varies. &
  $\binom{N}{3}$ &
  $\text{pos}(\text{3rd }A) - \text{pos}(\text{1st }A) - 1$ &
  \texttt{\{AAIAA, AZAAA, \ldots\}} \\
\textbf{Medium} &
  Positions of 1st and 4th $A$ fixed; inner positions vary. &
  $\binom{N}{2}$ &
  $\binom{\text{pos}(\text{4th }A)-\text{pos}(\text{1st }A)-1}{2}$ &
  \texttt{\{AZAAA, AAAZA, \ldots\}} \\
\textbf{Large} &
  Position of 1st $A$ fixed only. &
  $\binom{N}{1}$ &
  $\binom{N - 2 - \text{pos}(\text{1st }A)}{3}$ &
  \texttt{\{AZAAA, AAAAI, \ldots\}} \\
\bottomrule
\end{tabular}}
\end{table*}

\subsection{Intra-Group Scheduling}
\label{sec:intra-atomic}

Singleton groups contain exactly one term and require no intra-group
scheduling. We focus on scheduling within Atomic and Mini groups; the
scheduling within Medium and Large groups follows the same logic as
inter-group scheduling for Mini groups.

\paragraph{Intra-Atomic-Group Scheduling.}
To execute a single Pauli string within an atomic group, we select one
$A$ operator as the root and construct a CNOT chain linking all relevant
qubits to this root. A naive approach follows the original qubit layout
(Figure~\ref{fig:GroupGate}(a)). Because the 2nd and 3rd $A$ operators
are often separated by a long stretch of identity operators, realizing
the required connectivity naively demands many SWAP gates. Since all
eight strings in an atomic group share the same positional pattern, this
SWAP overhead is incurred repeatedly.

To eliminate this overhead, we apply \emph{qubit remapping} before
executing each atomic group. We reorder the physical qubits so that
those corresponding to $A$ operators, $Z$ operators, and $I$ operators
are each placed in contiguous regions (Figure~\ref{fig:GroupGate}(b)).
Under this layout, the required CNOT connections can be realized without
any SWAP gates during string execution.

Remapping also enables gate cancellation across strings within the same
atomic group. When two strings are executed consecutively, the CNOT and
basis-change gates (H and S) acting on the shared portion of the
connectivity path from leaf to root are identical and cancel directly.
By ordering strings so that consecutive pairs share the longest possible
common suffix, we achieve $O(N)$ gate cancellation per transition within
an atomic group. For example, strings $XX\alpha$ and $YX\alpha$ (where
$\alpha$ denotes a shared suffix) can be scheduled adjacently to cancel
all gates acting on $\alpha$.

\begin{figure}[t]
    \centering
    \includegraphics[width=0.49\textwidth]{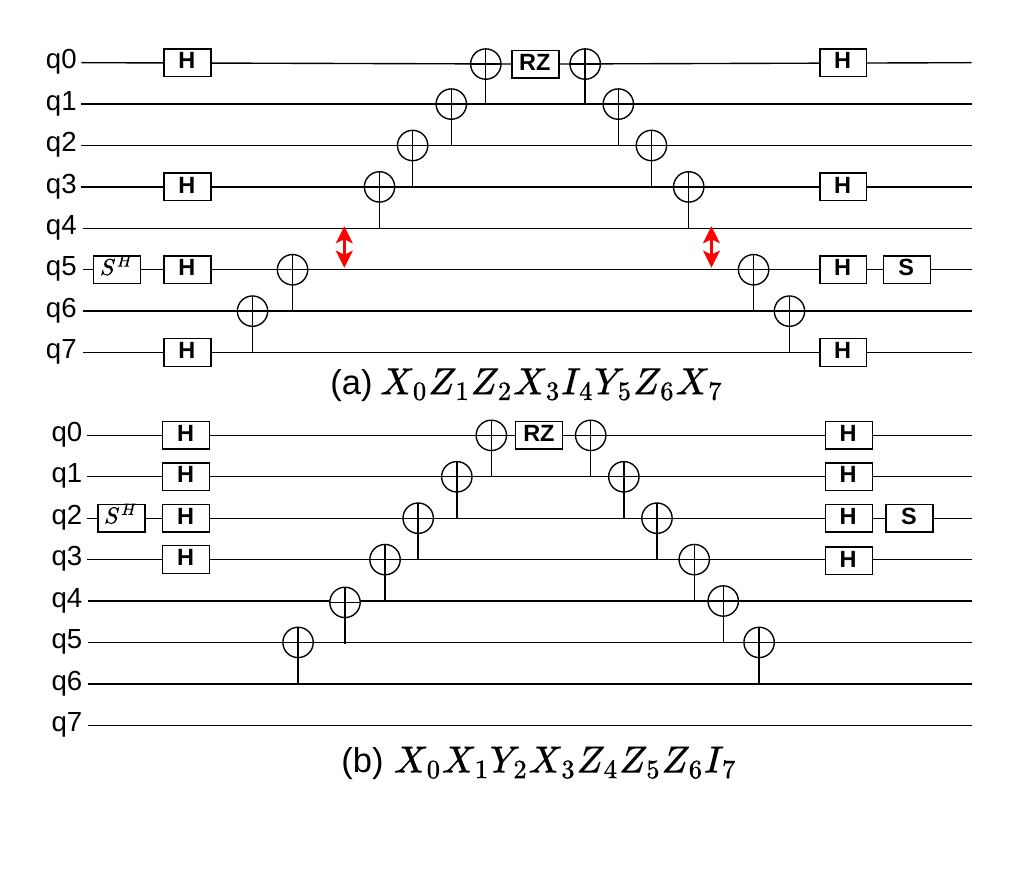}
    \vspace{-0.1in}
    \caption{(a) Naive qubit layout for an atomic group, requiring
    multiple SWAP gates to bridge non-adjacent $A$ operators. (b)
    \sysname{}'s remapped layout, clustering $A$, $Z$, and $I$ operators
    into contiguous regions, eliminating SWAP overhead.}
    \label{fig:GroupGate}
\end{figure}

\begin{figure}[t]
    \centering
    \includegraphics[width=0.49\textwidth]{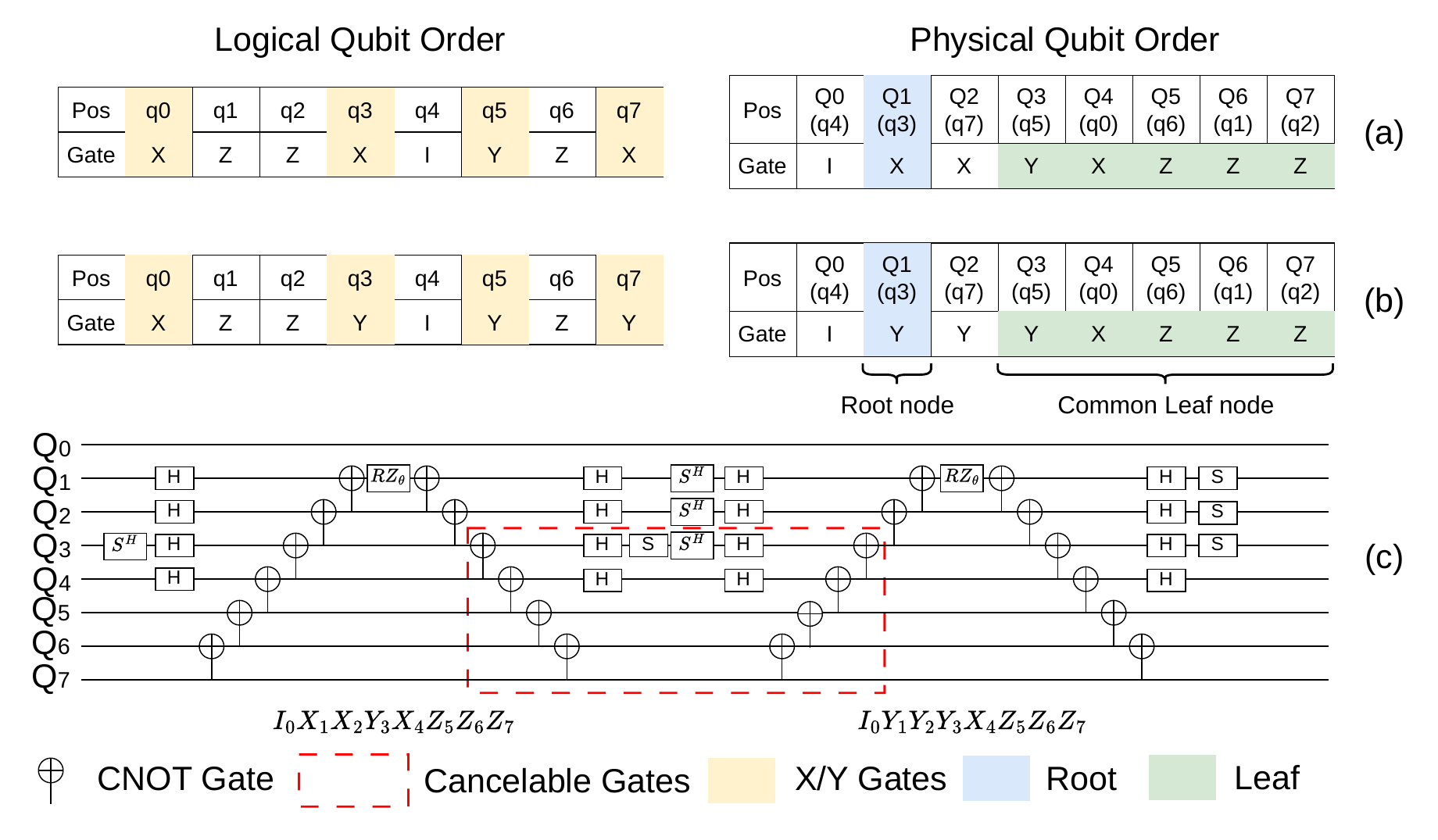}
    \vspace{-0.1in}
    \caption{Gate cancellation between two singleton groups (a) and (b)
    that share common leaf qubits. All gates on the shared path from
    leaf nodes to the root (Q1) cancel when the groups are scheduled
    consecutively.}
    \label{fig:OrderWithinAtomicGroup}
\end{figure}

\begin{figure*}[t]
    \centering
    \includegraphics[width=\textwidth]{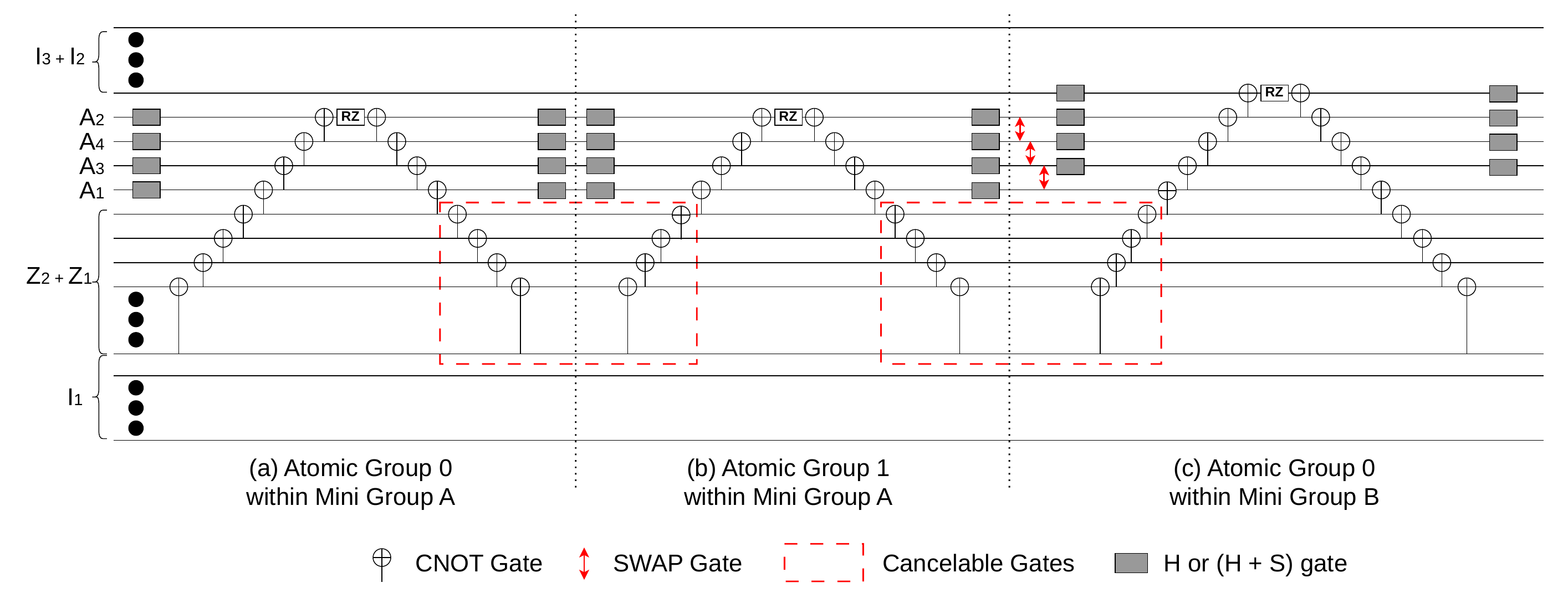}
    \vspace{-0.3in}
    \caption{Illustration of CNOT and SWAP gate cancellation across
    atomic groups within a mini group ((a)$\to$(b)) and across adjacent
    mini groups ((b)$\to$(c)).}
    \label{fig:GateCancel}
\end{figure*}

\paragraph{Intra-Mini-Group Scheduling.}
Within a mini group, all strings share the positions of the 1st, 3rd,
and 4th $A$ operators; only the position of the 2nd $A$ operator varies.
We define two mini groups as \emph{adjacent} if the positions of their
2nd $A$ operators differ by one.

Before executing a mini group, we apply a global qubit remapping so that
all strings conform to the canonical layout $I^*AAAAZ^*I^*$, in which
all $A$ operators are clustered contiguously, followed by $Z$ operators
and then identity operators. Under this layout, transitioning between
two adjacent mini groups requires shifting the 2nd $A$ operator by one
position, which can always be accomplished with exactly 3 SWAP
operations. This transformation preserves the $I^*AAAAZ^*I^*$ pattern,
keeping subsequent executions compatible with our intra-group strategy.

This canonical layout also enables CNOT gate cancellation across
adjacent mini groups. As illustrated in Figure~\ref{fig:GateCancel},
two cases arise. For atomic groups within the same mini group
(Figure~\ref{fig:GateCancel}(a)$\to$(b)), all strings share identical
positions, so the cancellation mechanism from intra-atomic scheduling
applies directly. For atomic groups across adjacent mini groups
(Figure~\ref{fig:GateCancel}(b)$\to$(c)), the 3 SWAP operations
required for the transition are confined to the qubits hosting the four
$A$ operators. The CNOT chains acting on the $Z$-operator qubits are
unaffected and can therefore be largely canceled across group
boundaries.

\subsection{Inter-Group Scheduling}
\label{sec:inter-atomic}

Our overall scheduling procedure iterates over all groups
hierarchically, executing mini groups in sequence according to
Algorithm~\ref{algo:all-to-all-synthesis}. Each group corresponds to a
specific qubit permutation of length $N$ determined by the first Pauli
string in the group. Transitioning between groups is equivalent to
transforming one permutation into another, which we accomplish using
parallel odd-even transposition sorting~\cite{habermann1972}. This
method realizes any permutation using only nearest-neighbor SWAP
operations in $O(N)$ depth, providing a uniform and deterministic
scheduling strategy for all inter-group transitions.

\begin{table}[t]
\caption{Logical qubit layout for executing a Pauli string with $A$
operators at positions $p$, $q$, $r$, and $s$.}
\label{table:PauliLogic}
\resizebox{\columnwidth}{!}{
\begin{tabular}{lccccccccc}
\toprule
\textbf{Region} & $I_1$ & $A_1$ & $Z_1$ & $A_2$ & $I_2$
               & $A_3$ & $Z_2$ & $A_4$ & $I_3$ \\
\midrule
Start & 0   & $p$   & $p+1$ & $q$   & $q+1$ & $r$   & $r+1$ & $s$   & $s+1$ \\
End   & $p-1$ & $p$ & $q-1$ & $q$   & $r-1$ & $r$   & $s-1$ & $s$   & $n-1$ \\
\bottomrule
\end{tabular}}
\end{table}

\begin{algorithm}[t]
\caption{Hierarchical scheduling for all-to-all Hamiltonians}
\label{algo:all-to-all-synthesis}
\begin{algorithmic}[1]
\For{each large group $L$}
    \State Initialize qubit mapping for $L$
    \For{each medium group $M \in L$}
        \State Adjust mapping to align the 4th $A$ operator
        \For{each mini group $m \in M$}
            \State Apply 3 SWAPs to achieve the $I^*AAAAZ^*I^*$ layout
            \For{each atomic group $a \in m$}
                \State Cancel adjacent Clifford gates
                \State Execute all Pauli strings in $a$
            \EndFor
        \EndFor
    \EndFor
\EndFor
\end{algorithmic}
\end{algorithm}

\subsection{Complexity Analysis}
\label{sec:complexity}

Algorithm~\ref{algo:all-to-all-synthesis} produces circuits with gate
count and circuit depth both bounded by $O(N^4)$. We provide a proof
sketch for circuit depth; the full gate count bound follows from a
charging argument detailed in Appendix~\ref{appendix:gate_count_bound}.

The total circuit depth has two components: intra-group and inter-group
scheduling. For intra-mini-group scheduling, atomic groups are executed
sequentially and transitioning between consecutive atomic groups
requires 3 additional SWAP stages. The gate cancellation enabled by our
scheduling strategy (Figure~\ref{fig:GateCancel}) ensures that each
atomic group contributes only $O(1)$ amortized depth, because the
$O(N)$ depth incurred by the first and last groups in each mini group is
amortized over $O(N)$ atomic groups. Since there are $O(N^4)$ atomic
groups in total, the total depth from intra-group scheduling is
$O(N^4)$.

For inter-group scheduling, the numbers of mini, medium, and large
groups are $\binom{N}{3}$, $\binom{N}{2}$, and $\binom{N}{1}$,
respectively (Table~\ref{table:GroupDef}). Each group transition
requires $O(N)$ depth via odd-even transposition sorting, contributing
$O(N^4)$ total depth across all transitions.

Summing both components gives $O(N^4)$ total circuit depth. Since the
number of second excitation terms grows as $\Theta(N^4)$, this bound is
asymptotically optimal: each Pauli string incurs only $O(1)$ amortized
depth. The overall circuit depth for the full electronic Hamiltonian---
including first excitation terms, whose count is dominated by the second
excitation terms---is therefore $O(N^4)$.

\subsection{Atomic Group Optimizations}

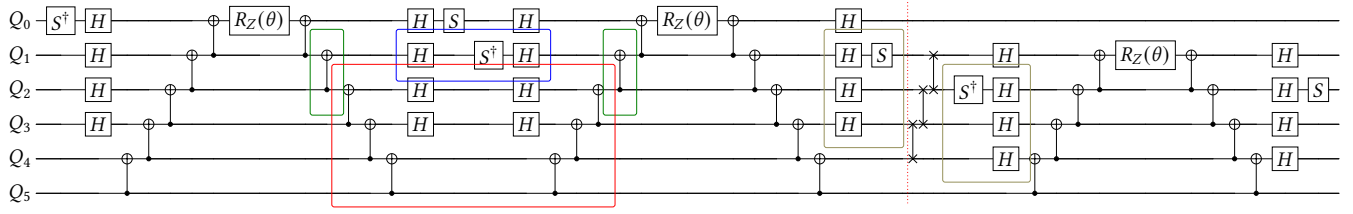
\begin{figure*}[t]
    \centering
\resizebox{\textwidth}{!}{
\begin{quantikz}[column sep=0.3cm, font=\huge,
  row sep={1cm,between origins}]
\lstick{$Q_{0}$} & \gate{S^\dagger} & \gate{H} & \qw & \qw & \qw & \qw & \targ{} & \gate{R_Z(\theta)} & \targ{} & \qw & \qw & \qw & \qw & \gate{H} & \gate{S} & \qw & \gate{H} & \qw & \qw & \qw & \qw & \targ{} & \gate{R_Z(\theta)} & \targ{} & \qw & \qw & \qw & \qw & \gate{H} & \qw & \qw\slice[style={dotted}]{} & \qw & \qw & \qw & \qw & \qw & \qw & \qw & \qw & \qw & \qw & \qw & \qw & \qw & \qw & \qw & \qw & \qw & \qw \\
\lstick{$Q_{1}$} & \qw & \gate{H} & \qw & \qw & \qw & \targ{} & \ctrl{-1} & \qw & \ctrl{-1} & \targ{} \gategroup[2,steps=1,style={draw=green!50!black, thick, rounded corners=2pt, inner sep=6pt}]{} & \qw & \qw & \qw & \gate{H} \gategroup[1,steps=4,style={draw=blue, thick, rounded corners=2pt, inner sep=6pt}]{} & \qw & \gate{S^\dagger} & \gate{H} & \qw & \qw & \qw & \targ{} \gategroup[2,steps=1,style={draw=green!50!black, thick, rounded corners=2pt, inner sep=6pt}]{} & \ctrl{-1} & \qw & \ctrl{-1} & \targ{} & \qw & \qw & \qw & \gate{H} \gategroup[3,steps=2,style={draw=yellow!50!black, thick, rounded corners=2pt, inner sep=6pt}]{} & \gate{S} & \qw & \qw & \qw & \swap{1} & \qw & \qw & \gate{H} & \qw & \qw & \qw & \targ{} & \gate{R_Z(\theta)} & \targ{} & \qw & \qw & \qw & \gate{H} & \qw & \qw \\
\lstick{$Q_{2}$} & \qw & \gate{H} & \qw & \qw & \targ{} & \ctrl{-1} & \qw & \qw & \qw & \ctrl{-1} & \targ{} \gategroup[4,steps=10,style={draw=red, thick, rounded corners=2pt, inner sep=6pt}]{} & \qw & \qw & \gate{H} & \qw & \qw & \gate{H} & \qw & \qw & \targ{} & \ctrl{-1} & \qw & \qw & \qw & \ctrl{-1} & \targ{} & \qw & \qw & \gate{H} & \qw & \qw & \qw & \swap{1} & \targX{} & \qw & \gate{S^\dagger} \gategroup[3,steps=2,style={draw=yellow!50!black, thick, rounded corners=2pt, inner sep=6pt}]{} & \gate{H} & \qw & \qw & \targ{} & \ctrl{-1} & \qw & \ctrl{-1} & \targ{} & \qw & \qw & \gate{H} & \gate{S} & \qw \\
\lstick{$Q_{3}$} & \qw & \gate{H} & \qw & \targ{} & \ctrl{-1} & \qw & \qw & \qw & \qw & \qw & \ctrl{-1} & \targ{} & \qw & \gate{H} & \qw & \qw & \gate{H} & \qw & \targ{} & \ctrl{-1} & \qw & \qw & \qw & \qw & \qw & \ctrl{-1} & \targ{} & \qw & \gate{H} & \qw & \qw & \swap{1} & \targX{} & \qw & \qw & \qw & \gate{H} & \qw & \targ{} & \ctrl{-1} & \qw & \qw & \qw & \ctrl{-1} & \targ{} & \qw & \gate{H} & \qw & \qw \\
\lstick{$Q_{4}$} & \qw & \qw & \targ{} & \ctrl{-1} & \qw & \qw & \qw & \qw & \qw & \qw & \qw & \ctrl{-1} & \targ{} & \qw & \qw & \qw & \qw & \targ{} & \ctrl{-1} & \qw & \qw & \qw & \qw & \qw & \qw & \qw & \ctrl{-1} & \targ{} & \qw & \qw & \qw & \targX{} & \qw & \qw & \qw & \qw & \gate{H} & \targ{} & \ctrl{-1} & \qw & \qw & \qw & \qw & \qw & \ctrl{-1} & \targ{} & \gate{H} & \qw & \qw \\
\lstick{$Q_{5}$} & \qw & \qw & \ctrl{-1} & \qw & \qw & \qw & \qw & \qw & \qw & \qw & \qw & \qw & \ctrl{-1} & \qw & \qw & \qw & \qw & \ctrl{-1} & \qw & \qw & \qw & \qw & \qw & \qw & \qw & \qw & \qw & \ctrl{-1} & \qw & \qw & \qw & \qw & \qw & \qw & \qw & \qw & \qw & \ctrl{-1} & \qw & \qw & \qw & \qw & \qw & \qw & \qw & \ctrl{-1} & \qw & \qw & \qw
\end{quantikz}
}
    \caption{Red: This portion of the circuit can be directly cancelled. Blue: It is known that $HSH=X^{1/2}$. Green: These CNOTs commute through $X$ on their target, thus they can also be cancelled. Yellow: These $H$, $S$, and $S^\dagger$ gates can be commuted through the $SWAP$ gates and cancelled.}
    \label{fig:OptCancel}
\end{figure*}

We can make additional adjustments to the order which strings within an atomic group are scheduled to cancel even more gates (Figure~\ref{fig:OptCancel}).

From the qubit order, the first and third factor of a Pauli string show up in the lower half of the basis rotations.
If they are the same Pauli, they can be cancelled (red rectangle in Figure~\ref{fig:OptCancel}).
Thus we pick the following order in (Eq~\ref{eq:string-ord}) to schedule strings in an atomic group.
This order maximizes the how often the first and third factor coincide (the underlined Paulis).
\begin{align}
\label{eq:string-ord}
\underline YX\underline YY,
\underline YY\underline YX,
\underline YX\underline XX,
\underline YY\underline XY,\nonumber\\
\underline XY\underline XX,
\underline XX\underline XY,
\underline XX\underline YX,
\underline XY\underline YY
\end{align}

This can be further improved by observing $HSH=X^{1/2}$ and $HS^\dagger H=X^{-1/2}$ (blue rectangle in Figure~\ref{fig:OptCancel}).
A CNOT applies $I$ or $X$ to the target qubit based on the control qubit's value, thus CNOT commutes with $X^{1/2}$ on the target qubit.
This allows us to cancel an additional pair of CNOTs per string (green rectangle in Figure~\ref{fig:OptCancel}).

Finally, when switching between atomic groups, we note that $A_1$, $A_3$, and $A_4$ remain on the same qubit.
Thus we can avoid doing any basis transformations here if we reverse the order which we schedule strings in an atomic group each iteration.
Equivalently, we can see this as commuting $H$, $S$, and $S^\dagger$ gates through the $SWAP$ gates and cancelling them (yellow rectangle in Figure~\ref{fig:OptCancel}).

\subsection{Depth Optimizations}
Observe that a Pauli string with index $(p,q,r,s)$ only interacts on qubits $p$ through $s$.
Thus if we have two Pauli strings with indices $(p,q,r,s)$ and $(p', q',r',s')$ such that $s<p'$, they operate on disjoint sets of qubits.
This means they can be run in parallel.

To make this precise, we construct a gadget $G(p,s)$ for each medium group (1st and 4th A are fixed).
This gadget has a pre-condition and post-condition that the supported qubits $p$ through $s$ are not permuted, and in the $Z$ basis.
The gadget preserves this condition by scheduling all the strings in the mini and atomic groups, before applying a permutation to return the qubits to their original order.
Note that $G(p,s)$ and $G(p',s')$ can be run in parallel, so long as $s<p'$.
All that remains is to find a suitable order to execute all the gadgets $G(p,s)$ for $0\le p<s\le n-1$.

\sysname{} implements a greedy scheduler based off of the circuit depth and supported qubits of each gadget.
For the current time step, we record which gadgets are running on each qubit.
Then we iteratively schedule the gadget which has the largest circuit depth, but can be run in parallel with the currently running gadgets.
If no candidates exist, advance the current time step until a candidate is found.
An example of a resulting schedule is shown in Figure~\ref{fig:GadgetScheduler}

\begin{figure}
    \centering
    \includegraphics[width=0.5\textwidth]{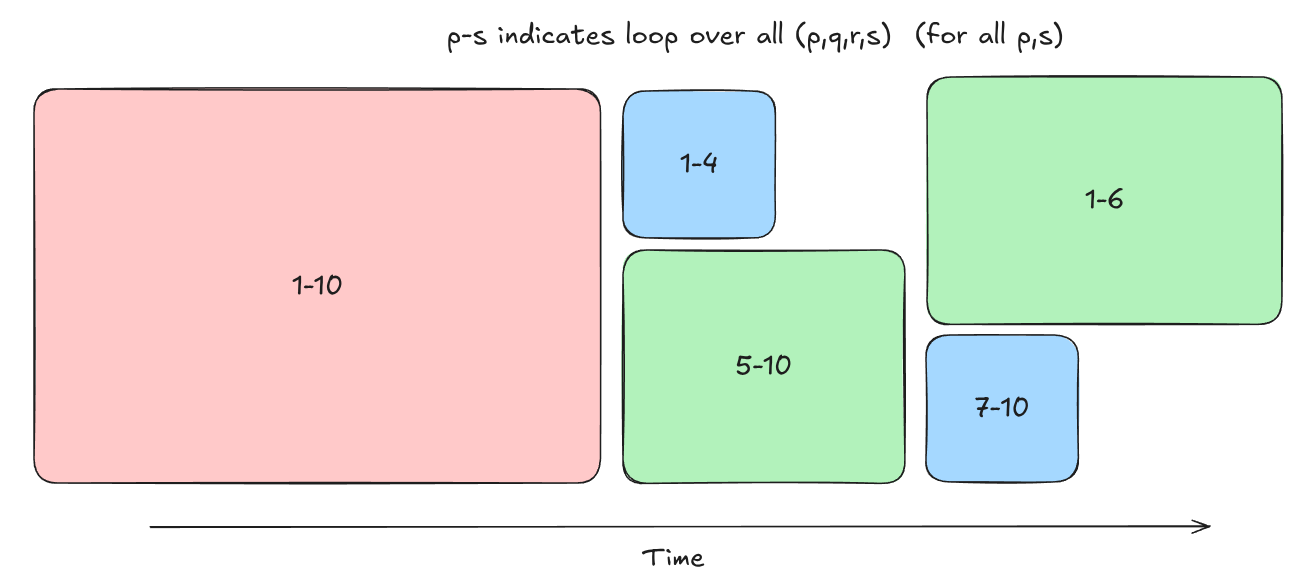}
    \caption{Each rectangle represents operations for a medium group.
    Medium groups are opportunistically scheduled, with the longest
    task being prioritized whenever free space is available.}
    \label{fig:GadgetScheduler}
\end{figure}

\section{Evaluation}
\label{sec:eval}

We evaluate \sysname{} by answering the following questions.
\begin{enumerate}[label=Q\arabic*)]
    \item How effectively does \sysname{} reduce circuit depth and gate
    count?
    \item How well do these metrics scale with problem size?
    \item How efficient is \sysname{}'s compilation time?
\end{enumerate}

\subsection{Experiment Setup}
\label{sec:experiment_setup}

\paragraph{Benchmarks.}
We evaluate on two benchmark classes. First, we use the UCCSD
ansatz~\cite{uccsd} instantiated on 15 distinct molecules retrieved from
PubChem~\cite{sunghwan2024pubchem} and constructed using
PySCF~\cite{PySCF}. Table~\ref{tab:uccsd} lists the molecules and
their ansatz properties. Second, we generate full-rank all-to-all
electronic Hamiltonians for $4 \le N \le 40$ qubits to evaluate
scalability when every possible Pauli string is included.

\begin{table}[htbp]
\caption{UCCSD ansatz properties by molecule.}
\label{tab:uccsd}
\begin{center}
\resizebox{\columnwidth}{!}{
\begin{tabular}{|c|c|c|c|c|c|}
\hline
\textbf{Molecule} & \textbf{Orbitals} & \textbf{Particles}
& \textbf{Qubits} & \textbf{\# 1st Exc.} & \textbf{\# 2nd Exc.} \\
\hline
LiH             & 6  & (2,2)   & 12 & 32   & 152     \\
BeH$_2$         & 7  & (3,3)   & 14 & 48   & 360     \\
CH$_4$          & 9  & (5,5)   & 18 & 80   & 1040    \\
MgH$_2$         & 11 & (7,7)   & 22 & 112  & 2072    \\
SiH$_4$         & 13 & (9,9)   & 26 & 144  & 3456    \\
CO$_2$          & 15 & (11,11) & 30 & 176  & 5192    \\
CH$_3$Cl        & 17 & (13,13) & 34 & 208  & 7280    \\
C$_2$F$_2$      & 20 & (15,15) & 40 & 300  & 15450   \\
H$_4$Si$_2$     & 22 & (16,16) & 44 & 384  & 25632   \\
C$_2$H$_4$F$_2$ & 24 & (17,17) & 48 & 476  & 39746   \\
CH$_3$ClS       & 26 & (21,21) & 52 & 420  & 30450   \\
COCl$_2$        & 28 & (24,24) & 56 & 384  & 25056   \\
C$_4$N$_2$      & 30 & (19,19) & 60 & 836  & 124982  \\
C$_3$H$_7$NO    & 32 & (20,20) & 64 & 960  & 165360  \\
C$_4$H$_4$N$_2$ & 34 & (21,21) & 68 & 1092 & 214578  \\
\hline
\end{tabular}}
\end{center}
\end{table}

\paragraph{Hardware Architectures.}
We evaluate on three modern quantum computing architectures: IBM
heavy-hex (Boston), square grid (Miami), and a linear architecture.
These represent a spectrum of connectivity patterns in current
superconducting hardware, spanning both NISQ and fault-tolerant regimes.
For Boston and Miami, \sysname{} first extracts a linear chain spanning
most or all qubits (Appendix~\ref{appendix:linear-pattern}) and applies
its compilation along this chain. The results therefore represent a
lower bound on \sysname{}'s full potential, since additional
connectivity beyond the extracted chain is not exploited.

\paragraph{Metrics.}
We measure total gate count, circuit depth, and compilation time. Gate
count and depth quantify hardware resource cost and critical path
length, respectively; SWAP gates are decomposed into 3 CNOT gates before
recording these metrics. Compilation time measures the wall-clock
duration to compile a set of Pauli strings to a hardware circuit,
excluding Hamiltonian construction and postprocessing shared by all
approaches. These metrics are standard in quantum circuit compilation
research~\cite{linearqft, li+:asplos19, molavi+:micro22,
zhang+:asplos21}.

\paragraph{Baselines.}
We compare against four baseline combinations: Paulihedral+JW,
Paulihedral+BK, Tetris+JW, and Tetris+BK. JW and BK represent
high and low Pauli weight mappings, respectively. HATT and Fermihedral
are excluded because they do not scale to the larger benchmarks and
because their advantage---reduced Pauli weight---is orthogonal to the
structural regularity that \sysname{} exploits. Tetris and Paulihedral
are selected as representative state-of-the-art hardware compilers for
VQE workloads. A uniform postprocessing pass---self-inverse gate
cancellation, single-qubit gate merging, and SWAP decomposition into
CNOTs---is applied to all compilers.

\paragraph{Hardware Setup.}
All experiments were run on an HPC cluster with Intel Xeon Platinum 8358
(2.60~GHz) and Xeon Gold 6230R (2.10~GHz) processors, running Red Hat
Enterprise Linux 9 and CentOS Linux 7. Tetris's default lookahead
setting causes runtimes to exceed the cluster's 72-hour timeout for the
three largest UCCSD molecules and all all-to-all Hamiltonians; lookahead
is therefore disabled for those benchmarks.

\subsection{Q1: Gate Count and Circuit Depth Reduction}
\label{sec:q1}

Since no single baseline is best across all benchmarks
(Figure~\ref{fig:misc_comparison}), we report \sysname{}'s improvement
relative to the best-performing baseline for each individual benchmark.

\paragraph{Total Gate Count.}
CNOT gates dominate total gate count and are \sysname{}'s primary
optimization target. \sysname{} reduces CNOT count by $62.73\%$ on
Linear, $59.57\%$ on Boston, and $36.73\%$ on Miami
(Figure~\ref{fig:total}). As connectivity increases, the advantage
decreases because \sysname{}'s linear-chain strategy does not utilize
additional edges beyond the extracted path. These CNOT reductions
translate to total gate count reductions of $56.39\%$ (Linear),
$52.28\%$ (Boston), and $26.22\%$ (Miami). \sysname{} consistently
outperforms all baselines on all three architectures.

\sysname{} uses more SWAP gates than the baselines
(Table~\ref{tab:gate_count}). This is intentional: \sysname{} applies
qubit reordering via SWAPs before executing each group to establish a
favorable qubit layout, which enables substantially more CNOT
cancellation in subsequent steps. The net effect is a lower overall
CNOT count despite the higher SWAP usage, making this tradeoff
beneficial.

\begin{figure*}[t]
\centerline{\includegraphics[width=\textwidth]{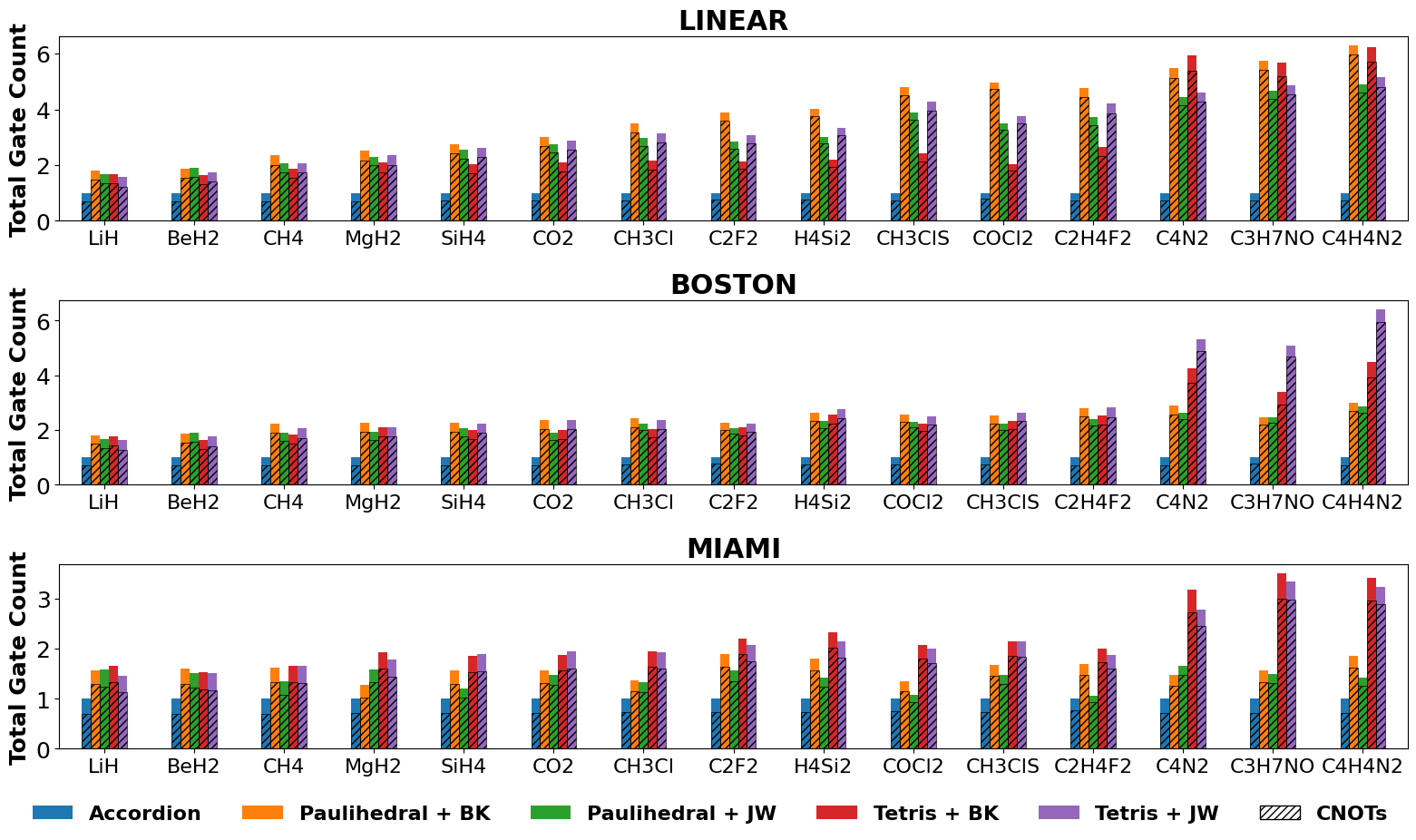}}
\vspace{-0.1in}
\caption{Total gate count on different molecules and architectures,
normalized to \sysname{} (= 1.0). CNOT gate count is overlaid; a full
breakdown is in Appendix~\ref{appendix:figures}.}
\label{fig:total}
\end{figure*}

\begin{figure*}[t]
\centerline{\includegraphics[width=\textwidth]{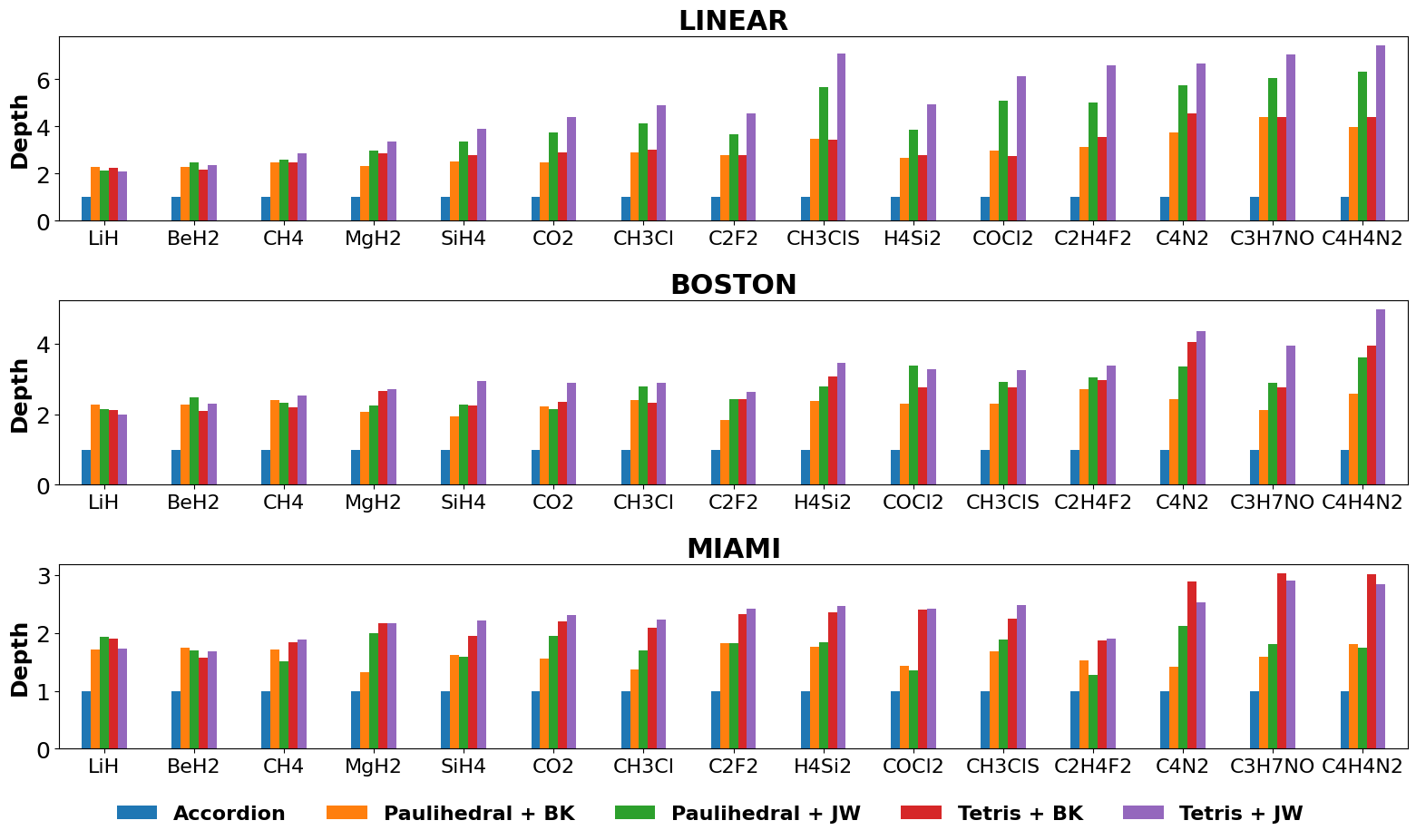}}
\vspace{-0.1in}
\caption{Circuit depth on different molecules and architectures,
normalized to \sysname{} (= 1.0).}
\label{fig:depth}
\end{figure*}

\begin{figure}[t]
    \centering
    \includegraphics[width=0.49\textwidth]{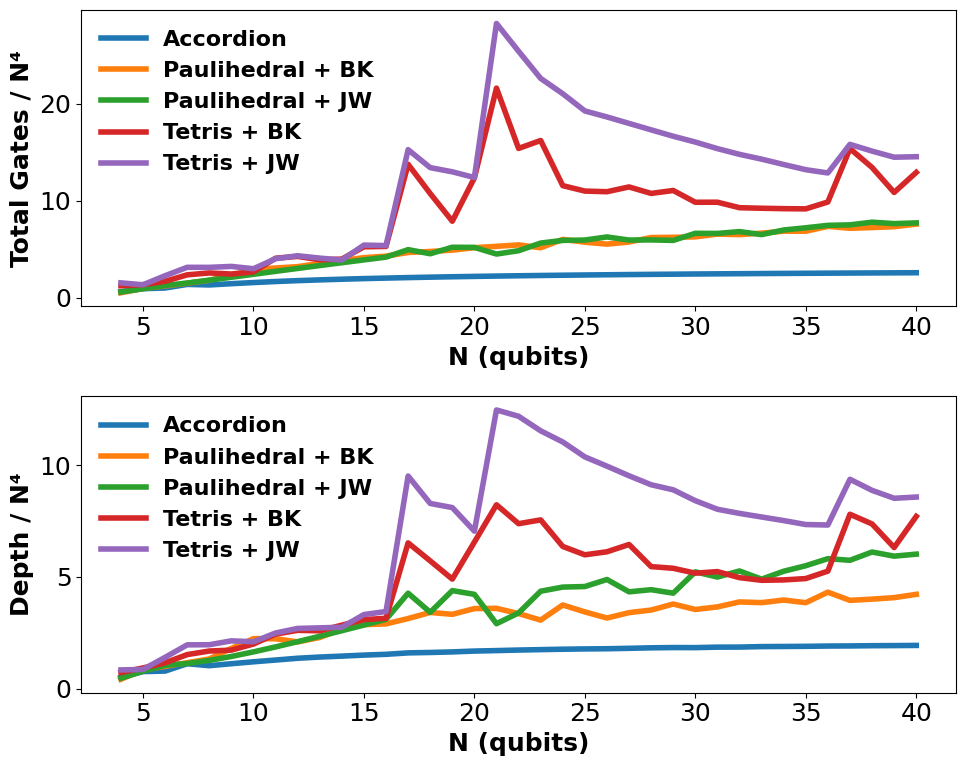}
    \vspace{-0.1in}
    \caption{Total gate count and circuit depth normalized by $N^4$ for
    all-to-all Hamiltonians on the Boston architecture.
    \sysname{} converges to a constant, confirming $O(N^4)$ complexity,
    while all baselines grow superquartically.}
    \label{fig:complete_n4}
\end{figure}

\paragraph{Single-Qubit Gates.}
\sysname{} introduces more single-qubit Clifford gates than the baselines
(Table~\ref{tab:gate_count}). This reflects a deliberate tradeoff:
Tetris and Paulihedral freely reorder strings across excitation terms to
maximize single-qubit gate cancellation, whereas \sysname{} fixes the
execution order within groups to maximize CNOT cancellation. In
fault-tolerant quantum computing, recent advances in magic state
cultivation~\cite{chen2026efficient, gidney2024magic} show that the cost
of non-Clifford gates (T gates) is comparable to that of two-qubit
Clifford gates (CNOTs). Since \sysname{} leaves the count of non-Clifford
gates (RZ gates) unchanged while substantially reducing CNOT count, it
provides the most meaningful savings in terms of overall fault-tolerant
resource cost. Furthermore, as Figure~\ref{fig:total} shows,
single-qubit Clifford gates constitute a small fraction of total gate
count in all cases.

\paragraph{Circuit Depth.}
\sysname{} reduces circuit depth by $64.14\%$ on Linear, $54.63\%$ on
Boston, and $34.82\%$ on Miami (Figure~\ref{fig:depth}). \sysname{}
consistenly outperforms all baselines on all architectures, but the
performance gap relative to baselines shrinks as the molecular benchmark
size increases. As with gate count, the advantage decreases with increasing
hardware connectivity, as the linear-chain abstraction does not exploit the
additional edges available in the Miami topology. Exploiting richer
hardware connectivity is a direction for future work.

\begin{table}[tbp]
\caption{Gate count breakdown for all compiler combinations on
C$_4$H$_4$N$_2$ (68 qubits) across hardware topologies. \sysname{}
produces identical circuits regardless of architecture, since it treats
all topologies as a linear chain.}
\label{tab:gate_count}
\centering
\resizebox{\columnwidth}{!}{
\begin{tabular}{|l|l|r|r|r|}
\hline
\textbf{Arch} & \textbf{Compiler} & \textbf{CNOT} & \textbf{Clifford 1Q} & \textbf{RZ} \\
\hline
\multirow{5}{*}{Linear} & \textsc{Tet}+JW & 42,603,032 & 2,261,252 & 859,404 \\
 & \textsc{Tet}+BK & 50,784,442 & 3,749,976 & 859,404 \\
 & \textsc{PH}+JW & 40,875,374 & 1,759,089 & 859,404 \\
 & \textsc{PH}+BK & 52,964,865 & 2,134,139 & 859,404 \\
 & \textsc{Accordion} & 6,399,354 & 1,633,367 & 859,404 \\
\hline
\multirow{5}{*}{Boston} & \textsc{Tet}+JW & 52,825,960 & 3,213,168 & 859,404 \\
 & \textsc{Tet}+BK & 34,749,240 & 4,217,391 & 859,404 \\
 & \textsc{PH}+JW & 23,311,396 & 1,285,155 & 859,404 \\
 & \textsc{PH}+BK & 23,847,989 & 2,002,719 & 859,404 \\
 & \textsc{Accordion} & 6,399,354 & 1,633,367 & 859,404 \\
\hline
\multirow{5}{*}{Miami} & \textsc{Tet}+JW & 25,663,888 & 2,168,896 & 859,404 \\
 & \textsc{Tet}+BK & 26,265,837 & 3,321,044 & 859,404 \\
 & \textsc{PH}+JW & 11,152,927 & 613,090 & 859,404 \\
 & \textsc{PH}+BK & 14,468,396 & 1,240,679 & 859,404 \\
 & \textsc{Accordion} & 6,399,354 & 1,633,367 & 859,404 \\
\hline
\end{tabular}
}
\end{table}

\begin{mdframed}[linewidth=0.8pt, innerleftmargin=6pt, innerrightmargin=6pt]
\textbf{Q1 Summary.} \sysname{} reduces both gate count and circuit
depth relative to the best baseline across all tested architectures.
The benefit is most pronounced on architectures with limited
connectivity, such as Linear and Boston.
\end{mdframed}

\subsection{Q2: Scalability}
\label{sec:q2}

We evaluate scalability using all-to-all Hamiltonians over a range of
qubit counts, since the number of second excitation terms grows
predictably as $\binom{N}{4} = O(N^4)$. Results are reported on Boston
as a representative real hardware architecture.

For any approach, at least $\binom{N}{4}$ RZ gates are required---one
per second excitation term---establishing $\Omega(N^4)$ as a lower
bound on gate count. Figure~\ref{fig:complete_n4} plots total gate count
and circuit depth normalized by $N^4$. For \sysname{}, both metrics
converge to a constant as $N$ grows, confirming $O(N^4)$ empirical
scaling consistent with the theoretical bound. All baseline approaches
exhibit superquartic growth, consistent with Figure~\ref{fig:ScalabilityA2A}.
While \sysname{} is not always optimal for small $N$, its advantage
grows monotonically with problem size---an important property given the
long-term trajectory toward larger quantum systems.

\begin{mdframed}[linewidth=0.8pt, innerleftmargin=6pt, innerrightmargin=6pt]
\textbf{Q2 Summary.} \sysname{} achieves $O(N^4)$ empirical scaling in
both gate count and circuit depth, matching the theoretical lower bound.
All baseline approaches scale superquartically.
\end{mdframed}

\subsection{Q3: Compilation Time}
\label{sec:q3}

Table~\ref{tab:runtime_comparison} reports compilation time for a
representative sample of benchmarks. \sysname{} completes all benchmarks
within 523 seconds and remains comparable with most benchmarks.
The only cases where \sysname{} trails behind are when compared against
Paulihedral+BK, or Tetris+BK.
Overall, \sysname{} suffers an average $15.89\%$ increased runtime, with
the worst case being $87.93\%$ increased runtime for BeH2.

A benefit of \sysname{}'s deterministic algorithm is an
architecture-independent compilation time. As shown in
Table~\ref{tab:runtime_comparison}, heuristic-based compilers such as
Paulihedral and Tetris can exhibit large runtime variation across
architectures for the same problem (e.g., Paulihedral+JW takes 3643
seconds on Miami versus 1516 seconds on Boston for C$_4$H$_4$N$_2$).
\sysname{}'s runtime is stable across architectures because it treats
all of them uniformly as a linear chain.

\begin{table*}[htbp]
\centering
\caption{Compiler runtime in seconds. Improv.\ = (best baseline $-$
\sysname{}) / best baseline. Entries marked $^*$ use Tetris with
lookahead disabled.}
\label{tab:runtime_comparison}
\vspace{-0.1in}
\resizebox{2\columnwidth}{!}{
\begin{tabular}{|l|l|r|rr|rr|rr|rr|}
\hline
\multirow{2}{*}{\textbf{Arch}} &\multirow{2}{*}{\textbf{Problem}} &\multicolumn{9}{c|}{\textbf{Runtime (s)}} \\
\cline{3-11}
 & & \textbf{Accordion} & \multicolumn{2}{c|}{\textbf{Paulihedral+BK}} & \multicolumn{2}{c|}{\textbf{Paulihedral+JW}} & \multicolumn{2}{c|}{\textbf{Tetris+BK}} & \multicolumn{2}{c|}{\textbf{Tetris+JW}} \\
\cline{4-5}\cline{6-7}\cline{8-9}\cline{10-11}
 & & & \textbf{Time (s)} & \textbf{Improv.} & \textbf{Time (s)} & \textbf{Improv.} & \textbf{Time (s)} & \textbf{Improv.} & \textbf{Time (s)} & \textbf{Improv.} \\
\hline
\multirow{5}{*}{MIAMI} & UCCSD MgH2 & 3.1 & 4.4 & $-$28.2\% & 3.8 & $-$17.8\% & 63.6 & $-$95.1\% & 62.7 & $-$95.0\% \\
 & UCCSD C2F2 & 28.9 & 29.4 & $-$1.8\% & 74.3 & $-$61.2\% & 6590.1 & $-$99.6\% & 4241.3 & $-$99.3\% \\
 & UCCSD C2H4F2 & 73.5 & 73.4 & $+$0.1\% & 251.6 & $-$70.8\% & 30359.5 & $-$99.8\% & 31434.2 & $-$99.8\% \\
 & Full Rank $N=40$ & 207.9 & 280.0 & $-$25.8\% & 772.3 & $-$73.1\% & 249.7 & $-$16.7\% & 300.6 & $-$30.8\% \\
 & UCCSD C4H4N2 & 523.1 & 446.1 & $+$17.3\% & 3643.6 & $-$85.6\% & 353.4 & $+$48.0\% & 709.6 & $-$26.3\% \\
\hline
\multirow{5}{*}{BOSTON} & UCCSD MgH2 & 3.1 & 4.6 & $-$31.6\% & 5.1 & $-$38.5\% & 64.2 & $-$95.1\% & 63.2 & $-$95.0\% \\
 & UCCSD C2F2 & 26.7 & 23.7 & $+$12.8\% & 64.8 & $-$58.7\% & 6543.5 & $-$99.6\% & 4242.3 & $-$99.4\% \\
 & UCCSD C2H4F2 & 79.8 & 59.8 & $+$33.3\% & 142.5 & $-$44.0\% & 30234.3 & $-$99.7\% & 31531.4 & $-$99.7\% \\
 & Full Rank $N=40$ & 209.0 & 229.4 & $-$8.9\% & 472.1 & $-$55.7\% & 256.9 & $-$18.6\% & 311.3 & $-$32.9\% \\
 & UCCSD C4H4N2 & 517.4 & 394.9 & $+$31.0\% & 1516.3 & $-$65.9\% & 400.3 & $+$29.3\% & 588.2 & $-$12.0\% \\
\hline
\multirow{5}{*}{LINEAR} & UCCSD MgH2 & 3.1 & 3.1 & $+$0.9\% & 3.6 & $-$13.7\% & 62.1 & $-$95.0\% & 61.1 & $-$94.9\% \\
 & UCCSD C2F2 & 26.4 & 21.7 & $+$21.8\% & 43.8 & $-$39.7\% & 4471.7 & $-$99.4\% & 4203.0 & $-$99.4\% \\
 & UCCSD C2H4F2 & 80.9 & 61.5 & $+$31.4\% & 137.6 & $-$41.2\% & 30234.9 & $-$99.7\% & 31509.0 & $-$99.7\% \\
 & Full Rank $N=40$ & 210.2 & 239.0 & $-$12.0\% & 463.4 & $-$54.6\% & 255.8 & $-$17.8\% & 280.3 & $-$25.0\% \\
 & UCCSD C4H4N2 & 518.4 & 584.5 & $-$11.3\% & 2243.0 & $-$76.9\% & 587.6 & $-$11.8\% & 801.2 & $-$35.3\% \\
\hline
\multicolumn{11}{l}{Improv.\ $= (\text{competitor} - \text{Accordion})/\text{competitor}$. Negative = Accordion is faster.} \\
\end{tabular}
}
\end{table*}

\begin{mdframed}[linewidth=0.8pt, innerleftmargin=6pt, innerrightmargin=6pt]
\textbf{Q3 Summary.} \sysname{} compiles all benchmarks within 523
seconds, remaining competitive compared to other compilers and
maintaining consistent runtime across hardware architectures.
\end{mdframed}

\section{Related Work}
\label{sec:related}

\paragraph{Multi-stage circuit generation for fermionic Hamiltonians.}
Prior work divides the circuit generation process into multiple
sequential stages. Fermion-to-qubit mappings---including
JW~\cite{JW}, BK~\cite{BK}, HATT~\cite{hatt}, and
Fermihedral~\cite{Fermihedral}---produce Pauli strings with varying
Pauli weights and structural properties. The Bonsai
algorithm~\cite{bonsai} generates tailored fermion-to-qubit mappings
derived from ternary trees. Given the resulting Pauli strings, circuit
synthesizers such as Tetris~\cite{Tetris} and
Paulihedral~\cite{Paulihedral} compile them into hardware circuits using
heuristic optimization. In contrast to this stage-wise design,
\sysname{} unifies the entire process by fixing the JW mapping and
developing tailored circuit synthesis algorithms that exploit JW's
structural regularity to reduce gate count, circuit depth, and
compilation time.

\paragraph{General-purpose quantum compiler optimizations.}
State-of-the-art quantum compilers such as Qiskit~\cite{qiskit},
Cirq~\cite{cirq}, and Braket~\cite{braket} incorporate optimization
passes including gate cancellation~\cite{nam2018automated} and qubit
routing~\cite{pozzi2022using}. These passes operate on arbitrary
circuits and can be applied to the output of any compilation pipeline,
including \sysname{}'s. The optimizations in \sysname{} are
complementary to these general-purpose compiler passes.

\section{Conclusion}
\label{conclusion}

We presented \sysname{}, an end-to-end framework for compiling
electronic structure fermionic Hamiltonians to hardware quantum circuits.
The central insight of \sysname{} is that the Jordan--Wigner mapping,
despite having the highest Pauli weight among standard mappings, produces
Pauli operators with a structural regularity that enables a class of
gate cancellation and qubit scheduling optimizations inaccessible to
general-purpose synthesizers. By co-designing the fermion-to-qubit
mapping with the with the circuit synthesis and hardware routing stages,
\sysname{} unlocks global optimization opportunities that stage-isolated
approaches cannot achieve.

Our key theoretical contribution is a proof that \sysname{} compiles
full-rank all-to-all electronic structure Hamiltonians into circuits with
$O(N^4)$ gate count and circuit depth, matching the information-theoretic
lower bound imposed by the $\Theta(N^4)$ number of second excitation
terms. This bound is achieved by a hierarchical Pauli string grouping
strategy that exploits JW's regular operator structure to maximize CNOT
cancellation across groups, combined with a qubit scheduling algorithm
that limits inter-group transition overhead to $O(1)$ amortized depth
per Pauli string. All existing mapper--synthesizer combinations exhibit
superquartic scaling in practice, making \sysname{} the first approach
to achieve asymptotically optimal circuit complexity for this problem
class.

Empirically, \sysname{} reduces gate count by up to $79.56\%$
(avg.\ $44.96\%$), circuit depth by up to $77.24\%$ (avg.\ $51.20\%$),
while maintaining competitive compilation time relative to
the best available baseline, across a suite of 15 molecular UCCSD
benchmarks on IBM heavy-hex, square grid, and linear architectures.
These improvements grow with problem size, a property that will become
increasingly important as quantum hardware scales toward the regime where
quantum chemistry simulations of practical interest become tractable.

More broadly, \sysname{} demonstrates that the conventional wisdom of
minimizing Pauli weight as a proxy for circuit quality is fundamentally
misguided. End-to-end circuit efficiency depends not only on the
compactness of intermediate representations but on the structural
properties that downstream compilation stages can exploit. We anticipate
that this co-design philosophy---fixing a structured intermediate
representation and specializing all downstream compilation to it---will
generalize beyond electronic structure Hamiltonians to other classes of
quantum simulation workloads.

Future work includes extending \sysname{} to exploit richer hardware
connectivity beyond the linear-chain abstraction, adapting the grouping
and scheduling algorithms to Fermi--Hubbard and other lattice
Hamiltonians, and integrating \sysname{}'s structured compilation with
fault-tolerant resource estimation frameworks.

\bibliographystyle{ACM-Reference-Format}
\bibliography{sample-base}

\balance
\appendix

\section{Linear pattern across topologies}
\label{appendix:linear-pattern}

\emph{Linear connectivity patterns} 
are prevalent across modern quantum hardware architectures. 
Despite differences in physical layouts, such as heavy-hex, linear chains, and 2D grids, 
we can always identify a \emph{logical line} that passes through most, if not all, qubits.
Figure~\ref{fig:linearCommonPattern} illustrates how such linear patterns can be extracted 
from various hardware topologies.

\begin{figure}[h]
    \centering
    \includegraphics[width=0.5\textwidth]{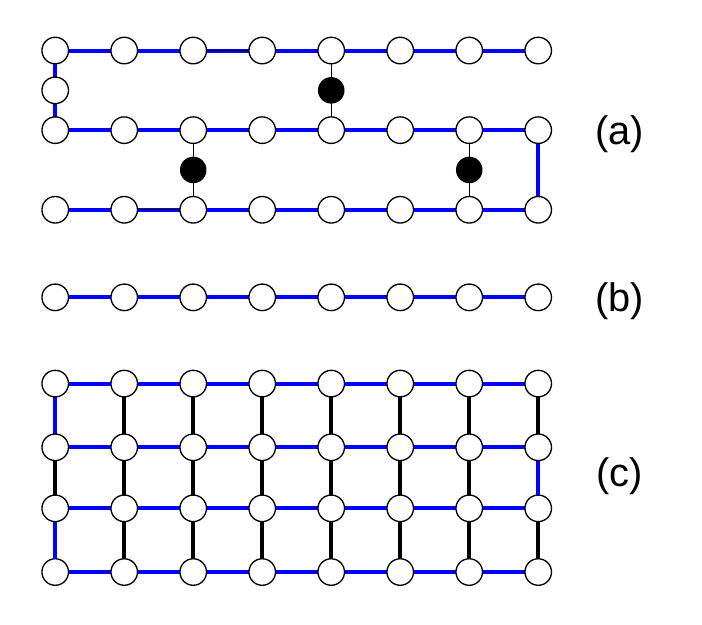}
    \vspace{-0.1in}
    \caption{We can always find a line (\textcolor{blue}{---}) passing most of the qubits in modern quantum hardware architectures, where (a) is Boston, (b) is Linear, and (c) is Miami.}
    \label{fig:linearCommonPattern}
\end{figure}

\begin{figure*}[t]
\centerline{\includegraphics[width=0.95\textwidth]{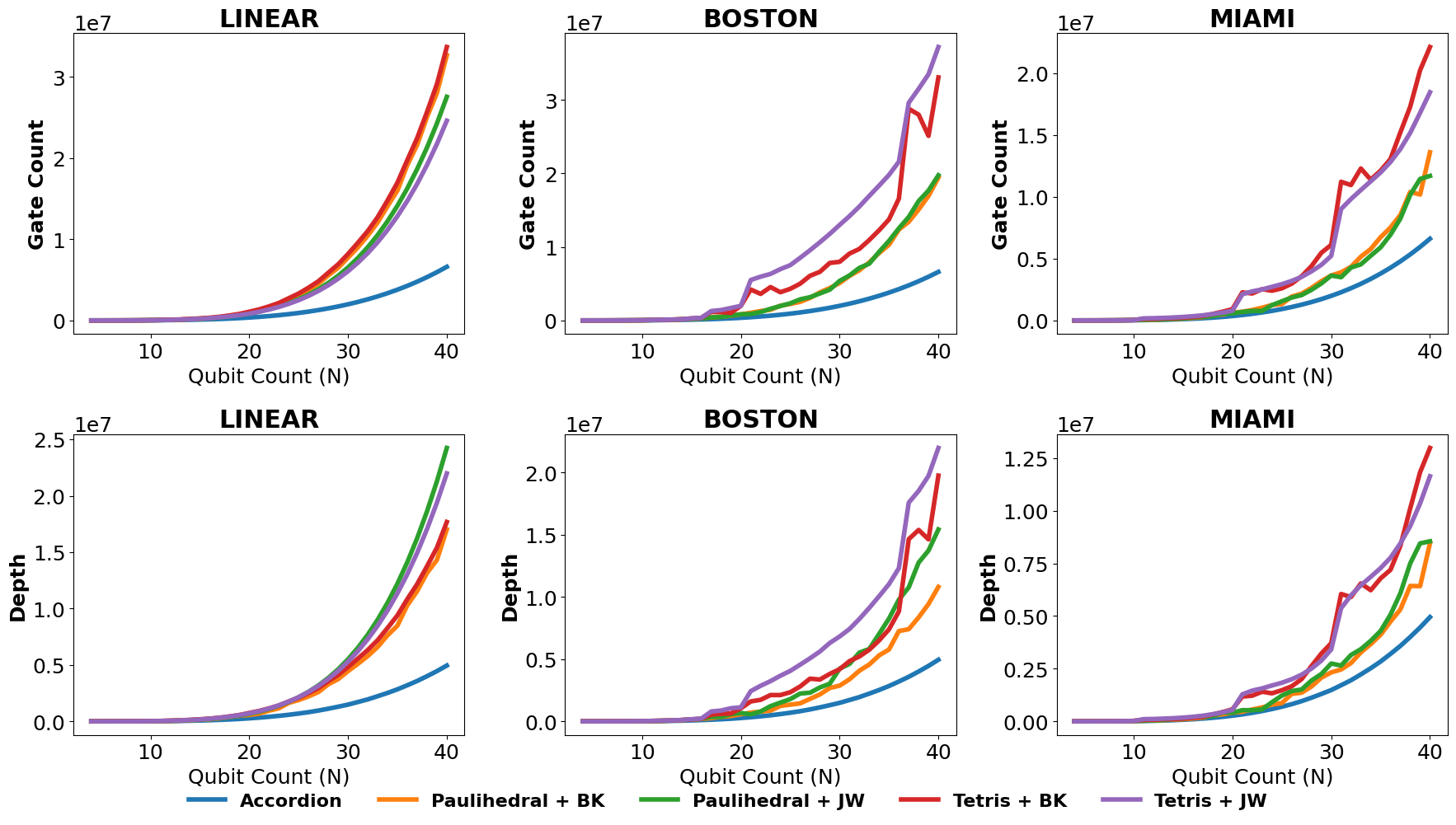}}
\caption{Trend of gate count and circuit depth usage among all approaches when increasing \# qubits in all-to-all Hamiltonians}
\label{fig:all-to-all-scale}
\end{figure*}



\section{Scalability for gate count and circuit depth}
\label{appendix:scalability}
We generate various benchmarks for both all-to-all and lattice Hamiltonians by increasing the problem size.
As is shown in Figure~\ref{fig:all-to-all-scale}, our approach consistently achieves lower gate count and circuit depth 
than existing methods for both all-to-all and lattice Hamiltonians 
across most hardware topologies. 
An exception arises in the Miami architecture for lattice Hamiltonians, 
where our method results in slightly higher circuit depth. 
This is primarily because our current design does not explicitly exploit 
the additional connectivity provided by the architecture, 
instead relying on a linearized execution pattern. 
Effectively leveraging such extra connectivity to further reduce circuit depth 
is an interesting direction for future work.

\newpage
\section{Upper bound for total gate count in \sysname{}}
\label{appendix:gate_count_bound}

We consider the problem of compiling a sum of second excitation terms for $n$ qubits, indexed by
\begin{equation}
    T=\{(p,q,r,s):1\le p<q<r<s\le n\}.
\end{equation}
Applying the Jordan-Wigner mapping, operators with index $(p,q,r,s)\in T$
are sent to strings of the form
\begin{equation}
    I^*A_pZ^*A_qI^*A_rZ^*A_sI^*,
    \label{eq:pauli-string}
\end{equation}
where $A_p,A_q,A_r,A_s\in\{X,Y\}$, and each $A_i$ is in the $i$th index of this string.
We define $\sigma(p,q,r,s):T\to S_n$ to reorder the numbers $[1,n]$ to the intervals in Table~\ref{table:TargetPerm}.
Our approach, Algorithm~\ref{algo:iteration} iterates through indices $(p,q,r,s)\in T$ and rearranges the qubits at each step according to the permutation $\sigma(p,q,r,s)$.
This is done by applying SWAPs to adjacent qubits, and bubble-sorting from the current permutation to $\sigma(p,q,r,s)$.
Then we synthesize all Pauli strings corresponding to the current index.

\begin{algorithm}[h]
\caption{Implementation for \sysname{}}
\label{algo:iteration}
\begin{algorithmic}[1]
\For{$p \in \{1, \ldots, n-3\}$}
    \For{$s \in \{p+3, \ldots, n\}$}
        \For{$r \in \{p+2, \ldots, s-1\}$}
            \For{$q \in \{p+1, \ldots, r-1\}$}
                \State \textit{Reverse order of iteration for each \{p,s,r\}}
                \State Go to permutation $\sigma(p, q, r, s)$
                \State Schedule Pauli strings
            \EndFor
        \EndFor
    \EndFor
\EndFor
\end{algorithmic}
\end{algorithm}

\begin{table}[h]
\caption{Target permutation $\sigma(p,q,r,s)$ for qubit order}
\label{table:TargetPerm}
\resizebox{\columnwidth}{!}{
\begin{tabular}{lccccccccc}
\toprule
Pauli Physical & $\text{I}_{1}$ & $\text{Z}_{2}$ & $\text{Z}_{1}$ & $\text{A}_{1}$ & $\text{A}_{3}$ & $\text{A}_{4}$ & $\text{A}_{2}$ & $\text{I}_{2}$ & $\text{I}_{3}$ \\ \midrule
Start & 1 & r+1 & p+1 & p & r & s & q & q+1 & s+1 \\
End   & p-1 & s-1 & q-1 & p & r & s & q & r-1 & n \\ \bottomrule
\end{tabular}}
\end{table}

Before describing the exponential, we define some shorthand for quantum gates.
For integers $1\le a<b\le n$, we define a \textbf{CNOT ladder} $C_{a,b}=\prod_{i=a}^{b-1}CNOT_{i,i+1}$, supported on qubits $a$ through $b$.
If $b>a$, let $C_{a,b}=C_{b,a}^\dagger$.
It can be shown that
\begin{equation}
    C_{a,b}C_{b,c}=C_{a,c}.
\end{equation}

For any Pauli string (\ref{eq:pauli-string}) apply permutation $\sigma(p,q,r,s)$. Then let $a$ be the index of the leftmost $Z$ ($\sigma(r+1)$), and $b$ the index of the rightmost $A$ ($\sigma(q)$). Write $U_{pqrs}$ to be a product of at most 8 $H$ and $S$ gates, supported on qubits $\sigma(\{p,q,r,s\})$, rotating the Pauli string into the Z basis. We implement the exponential with 
\begin{equation}
\label{eq:pauli-exp}
U_{pqrs}^\dagger
C_{a,b}
RZ(\theta)
C_{b,a}
U_{pqrs}
\end{equation}

\begin{lemma}
\label{lem:swap}
Algorithm~\ref{algo:iteration} generates $O(n^4)$ SWAP gates for the $n$ qubit second excitation term synthesis problem.
\end{lemma}

\begin{proof}
To bound the number of SWAP gates, it suffices to charge each SWAP to one of the variables $p,q,r,s$.
We record every time $(p,q,r,s)$ changes, and count the number of SWAPs needed to reach the new permutation.
We note that each time we do this, the maximum number of SWAPs we can incur is $O(n^2)$, the same number of steps needed to sort a list of $n$ elements.

\textbf{Case 1 (p), (s):}
These two loop bodies both run at most $O(n^2)$ times, thus we can incur at most $O(n^4)$ SWAPs iterating over p and s.

\textbf{Case 2 (q):}
This loop body runs $O(n^4)$ times. Each new iteration, we send $\sigma(p,q,r,s)$ to $\sigma(p,q+1,r,s)$, which costs 3 SWAPs. This can be seen by looking at Table~\ref{table:TargetPerm} and observing this action simply requires moving $A_2$ to the start of $Z_1$. Therefore we incur $O(n^4)$ SWAPs iterating over q.

\textbf{Case 3 (r):}
This loop body runs $O(n^3)$ times.
Each new iteration, we send $\sigma(p,q,r,s)$ to $\sigma(p,q,r+1,s)$.
We do not need to reset $q$ to $p+1$, since we reverse the order of the 4th level loop each time it is run. 
This requires us to send $A_3$ to the right of $I_2$, and the left of $Z_2$ to the old position of $A_3$ (Table~\ref{table:TargetPerm}).
We are shifting two qubits left and right at most $n$ steps, costing $O(n)$ SWAPs for the total action.
We note that the backwards iteration is simply the inverse action, thus this bound will still hold.
Therefore we incur $O(n^4)$ SWAPs interating over r.

Combining all three cases, this algorithm will generate at most $O(n^4)$ SWAP gates.
\end{proof}

\begin{lemma}
\label{lem:cnot} Algorithm~\ref{algo:iteration} generates $O(n^4)$ CNOT gates for the $n$ qubit second excitation term synthesis problem.
\end{lemma}

\begin{proof}
Each CNOT ladder is located between two Pauli exponentials (discounting the first and last one), thus we can categorize these based off of which indices they lie between.

\textbf{Case 1, same index:}
These CNOTs lie between two different Pauli exponentials, thus we write the first rotation as $U_{pqrs}$ and second as $\bar U_{pqrs}$.
Recall that $b$ is the location of the rightmost $A$, thus $U_{pqrs}$ and $\bar U_{pqrs}$ are supported on qubits $b-3$ to $b$.
Split $C_{b,a}$ into $C_{b,b-4}C_{b-4,a}$ and $C_{a,b}$ into $C_{a,b-4}C_{b-4,b}$.
Cancel $C_{b-4,a}$ and $C_{a,b-4}$ by commuting them through $U_{pqrs}U^\dagger_{pqrs}$, as these are supported on disjoint qubits.

\begin{align}
&U_{pqrs}^\dagger
C_{a,b}
RZ(\theta)
C_{b,a}
U_{pqrs}   
\bar U_{pqrs}^\dagger
C_{a,b}
RZ(\theta)
C_{b,a}
\bar U_{pqrs}   
\nonumber\\
=&
\ \ldots
C_{b,b-4}
\cancel{C_{b-4,a}}
U_{pqrs}
\bar U_{pqrs}^\dagger
\cancel{C_{a,b-4}}
C_{b-4,b}
\ldots
\end{align}

We are left with at most $8$ CNOTs, thus each case here is charged $O(1)$ CNOTs per occurrence.
Each second excitation term has $2^4$ choices of $A$, thus this case is charged $O(n^4)$ times, bounding the number of CNOTs generated here by $O(n^4)$.

\textbf{Case 2, $(p,q,r,s)$ to $(p,q+1,r,s)$}
We make a similar argument as in Case 1, but shift some of the gates by one qubit.
Since the second Pauli exponential uses $q+1$ instead of $q$, we have CNOT ladder $C_{a,b+1}$ and $U_{pqrs}U_{p(q+1)rs}^\dagger$ is supported on qubits $b-3$ to $b+1$.
Furthermore, we have 3 SWAP gates from Case 2 of Lemma~\ref{lem:swap}, say $SW$ supported on $b-3$ to $b$.

\begin{align}
&
\ldots
RZ(\theta)
C_{b,a}
U_{pqrs}   
SW
U_{p(q+1)rs}^\dagger
C_{a,b+1}
RZ(\theta)
\ldots
\nonumber\\
=&
\ \ldots
C_{b,b-4}
\cancel{C_{b-4,a}}
U_{pqrs}
SW
U_{p(q+1)rs}^\dagger
\cancel{C_{a,b-4}}
C_{b-4,b+1}
\ldots
\end{align}

This case incurs 9 cost per charge, and is charged $O(n^4)$ times (once for each $(p,q,r,s)$.
Note an identical argument works if the iteration direction is reversed.
Thus we have at most $O(n^4)$ CNOTs between different indices.

\textbf{Case 3, remainder:}
This case can only be charged $O(n^3)$ times, as it only occurs when $(p,r,s)$ changes.
Each CNOT ladder consists of at most $n$ CNOTs, thus there are $O(n^4)$ CNOTs in this case.

All three cases combine generate at most $O(n^4)$ CNOTs, giving the desired bound.

\end{proof}

\begin{theorem}
\label{thm:gate-count}
Algorithm~\ref{algo:iteration} generate $O(n^4)$ quantum gates for the $n$ qubit second excitation term synthesis problem.
\end{theorem}

\begin{proof}
From \ref{eq:pauli-exp}, we see every exponential has at most 17 single-qubit gates, thus there will be $O(n^4)$ single-qubit gates in the final circuit. By Lemma~\ref{lem:swap} and Lemma~\ref{lem:cnot}, we generate at most $O(n^4)$ two-qubit gates. Thus in total, we will generate at most $O(n^4)$ quantum gates.
\end{proof}

\begin{corollary}
Algorithm~\ref{algo:iteration} generates a circuit of depth at most $O(n^4)$
\end{corollary}
\begin{proof}
This follows immediately from Theorem~\ref{thm:gate-count}.
\end{proof}

\onecolumn

\section{Additional Figures}
\label{appendix:figures}

\begin{figure}[h]
\centerline{\includegraphics[width=\textwidth]{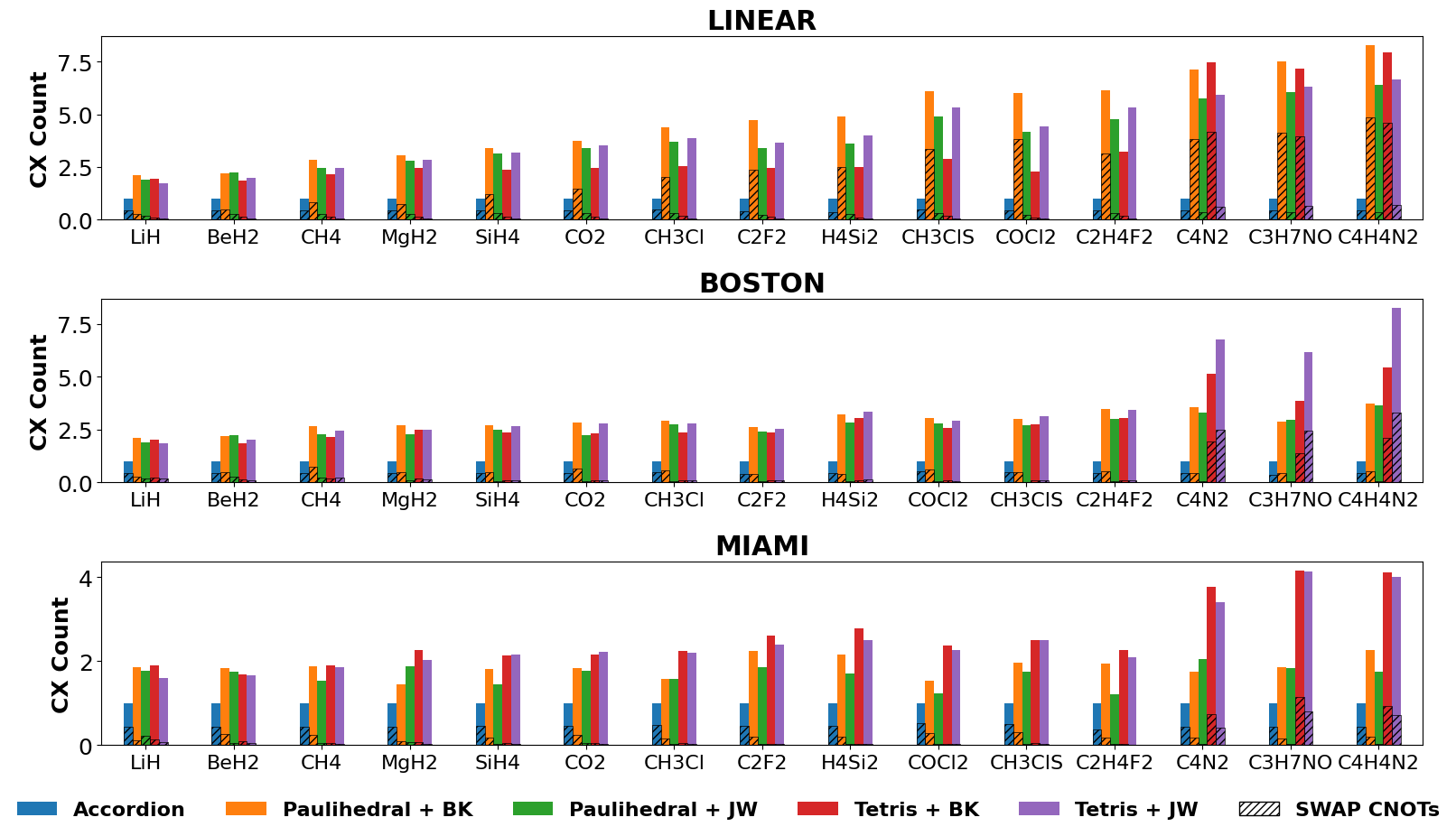}}
\caption{CNOT count for \sysname{}, Paulihedral, and Tetris with Jordan-Wigner (JW) and Bravyi-Kitaev (BK) qubit mapper on different molecules and architectures (normalized to \sysname{}). SWAP-induced CNOTs are overlayed.}
\label{fig:cx}
\end{figure}

\end{document}